\documentclass[a4paper,11pt]{article}
\usepackage[margin = 1in]{geometry}
\usepackage[T1]{fontenc}
\usepackage{amsmath,amssymb,amsthm}
\usepackage{thm-restate}
\usepackage[libertinus,mono=false]{newtx}
\usepackage{graphicx}
\usepackage[dvipsnames]{xcolor}
\usepackage[shortlabels]{enumitem}
\usepackage{algorithmicx}
\usepackage{algorithm}
\usepackage[noend]{algpseudocode}
\usepackage{xspace}
\usepackage{booktabs}
\usepackage[colorlinks=true, allcolors=NavyBlue]{hyperref}
\usepackage{cleveref}

\renewcommand{\arraystretch}{1.15}
\setlist{topsep=0.4\baselineskip,parsep=0pt,partopsep=0pt,itemsep=0.2\baselineskip}

\theoremstyle{plain}
\newtheorem{theorem}{Theorem}

\newtheorem*{claim*}{Claim}

\newtheorem{lemma}[theorem]{Lemma}
\newtheorem{corollary}[theorem]{Corollary}
\newtheorem{observation}[theorem]{Observation}
\theoremstyle{definition}
\newtheorem{definition}[theorem]{Definition}

\newtheorem{question}{Question}
\newenvironment{claimproof}[1][Proof of claim.]{\begin{proof}[#1]}{\end{proof}}

\mathcode`l="8000
\begingroup
\makeatletter
\lccode`\~=`\l
\DeclareMathSymbol{\lsb@l}{\mathalpha}{letters}{`l}
\lowercase{\gdef~{\ifnum\the\mathgroup=\m@ne \ell \else \lsb@l \fi}}%
\endgroup

\newcommand{\poly}{\mathrm{poly}}
\newcommand{\Reals}{\mathbb{R}}

\renewcommand{\leq}{\leqslant}
\renewcommand{\geq}{\geqslant}
\newcommand{\eps}{\varepsilon}

\newcommand{\card}[1]{\left\lvert #1 \right\rvert}
\newcommand{\norm}[1]{\left\lVert #1 \right\rVert}
\newcommand{\set}[1]{\left\{ #1 \right\}}
\newcommand{\class}[1]{\mathcal{#1}}

\newcommand{\dist}{\mathrm{dist}}
\newcommand{\proj}{\mathrm{proj}}
\newcommand{\grid}{\mathrm{grid}}
\newcommand{\portal}{\mathrm{port}}
\newcommand{\terminal}{\mathrm{term}}
\newcommand{\nxt}{\mathrm{next}}

\newcommand{\bd}{\partial}
\newcommand{\Otilde}{\widetilde{O}}
\newcommand{\OV}{\textsc{Orthogonal Vectors}\xspace}
\newcommand{\touring}{\textsc{Polygon Touring}\xspace}
\newcommand{\touringInt}{\textsc{Interval Touring}\xspace}
\newcommand{\touringOrtho}{\textsc{Orthogonal Polygon Touring}\xspace}
\newcommand{\touringOrthoParam}{\textsc{Orthogonal Polygon Touring$(n,k)$}\xspace}
\newcommand{\touringOrthoThreeParam}{\textsc{Orthogonal Polytope Touring$(n,k)$}\xspace}

\makeatletter
\def\@maketitle{%
\newpage
\null
\vskip 2em%

\begin{center}%
\let \footnote \thanks
    {\LARGE \@title \par}%
    \vskip 1.5em%
    {\large%
    \lineskip .5em%
    \renewcommand{\arraystretch}{2}%
    \begin{tabular}[t]{ccc}%
        \@author
    \end{tabular}\par}%
    \vskip 1em%
\end{center}%
\par
\vskip 1.5em}
\makeatother

\newcommand{\etal}{\emph{et~al.}}
\newcommand{\email}[1]{\href{mailto:#1}{\texttt{#1}}}
\newcommand{\entry}[4]{
    {\renewcommand{\arraystretch}{1}%
    \begin{tabular}[t]{c}
        #1\footnote{#2. Email: \email{#3}. #4}
    \end{tabular}}
}

\newcommand{\EMPH}[1]{\emph{\textcolor{BrickRed}{#1}}}

\title{Touring a Sequence of Orthogonal Polygons}
\author{%
    \entry{Katrin Casel}{Humboldt-University Berlin, Germany}{katrin.casel@hu-berlin.de}{} &
    \entry{Sándor Kisfaludi-Bak}{Aalto University, Espoo, Finland}{sandor.kisfaludi-bak@aalto.fi}{Supported by the Research Council of Finland, Grant 363444.} &
    \entry{Linda Kleist}{University of Hamburg, Germany}{linda.kleist@uni-hamburg.de}{} \\
    \entry{Jeroen S.K. Lamme}{Eindhoven University of Technology, the Netherlands.}{j.s.k.lamme@tue.nl}{} &
    \entry{Eunjin Oh}{POSTECH, South Korea}{eunjin.oh@postech.ac.kr}{Supported by Institute of Information \& Communications Technology Planning \& Evaluation (IITP) grant funded by the Korea government (MSIT) (No.RS-2024-00440239, Sublinear Scalable Algorithms for Large-Scale Data Analysis) and the National Research Foundation of Korea (NRF) grant funded by the Korea government (MSIT) (No.RS-2024-00358505).} &
    \entry{Yanheng Wang}{ETH Zürich, Switzerland}{yanheng.wang@inf.ethz.ch}{Part of this work was finished at Saarland University, Saarbrücken, Germany. It was part of the project TIPEA that has received funding from the European Research Council (ERC) under the European Unions Horizon 2020 research and innovation programme (grant agreement No. 850979).}
}
\date{}

\begin{document}
\maketitle
\thispagestyle{empty}

\begin{abstract}
    We study the problem of computing a shortest tour that visits a sequence of $k$ polygons $P_1,\dots, P_k$ with a total number of $n$ vertices. A tour is an oriented curve such that there exist points $p_i\in P_i$ for all $i$ where $p_i$ appears not after $p_{i+1}$.
    In a seminal paper, Dror, Efrat, Lubiw and Mitchell (STOC 2003) considered the problem under $L_2$ distance, and gave $\widetilde O(nk)$ and $\widetilde O(nk^2)$ algorithms for disjoint and intersecting convex polygons, respectively.
    
    In this paper, we consider the orthogonal setting (with orthogonal polygons and Manhattan distance) and obtain the following results:
    \begin{itemize}
        \item a truly subquadratic $\widetilde O(n^{2-\frac{1}{48}})$ algorithm when consecutive polygons in the sequence are disjoint;
        \item an $\widetilde O(n)$ algorithm for ortho-convex polygons when consecutive polygons are disjoint;
        \item an $O(n)$ algorithm for axis-aligned rectangles;
        \item $\widetilde O(n^2)$ and $\widetilde O(n^{1.5}k^2)$ algorithms without restrictions.
    \end{itemize}
    Our algorithms build on a wide range of techniques, including additively weighted Voronoi diagrams, rectangle decompositions, persistent data structures, and dynamic distance oracles for weighted planar graphs.
\end{abstract}

\clearpage
\setcounter{page}{1}
\section{Introduction}

Shortest paths between two points in geometric environments can often be computed in near-linear time. However, in settings where the path is constrained to avoid or pass through given regions, polynomial-time algorithms might not exist. For example, computing a shortest path that avoids axis-aligned half-planes in $\Reals^3$, or one that passes through non-convex regions in $\Reals^2$ are both NP-hard problems~\cite{MitchellS04,DELM03}.

In this paper we study a constrained shortest path problem \touring: given a sequence $P_1,\dots,P_k\subset \Reals^2$ of polygons with $n$ vertices in total, compute (the length of) a shortest tour that visits the polygons in order. A \EMPH{tour} $\pi$ is an oriented open curve, and it \EMPH{visits} a sequence of polygons $P_1,\dots,P_k$ if there exist points $p_i \in P_i$ for all $i \in \set{1, \dots, k}$ such that $\pi$ goes through $p_1,\dots,p_k$ in order. We highlight that the points may coincide; see \Cref{fig:intro} for an example. One can measure the length of a tour in different metrics, and natural choices include the Manhattan ($L_1$) and Euclidean ($L_2$) metrics.

\begin{figure}[htb]
    \centering
    \includegraphics{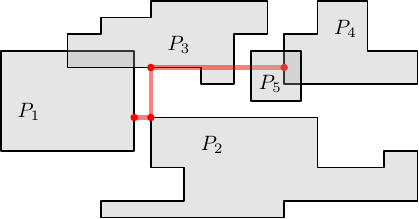}
    \caption{A tour visiting a sequence of five orthogonal polygons. All polygons but $P_2$ are ortho-convex. Note that deleting $P_5$ yields a step-disjoint instance.
    }
    \label{fig:intro}
\end{figure}

\touring has obvious applications in motion planning and logistics in physical environments~\cite{FaiglVP11,IshidaRH19,KulichVM23,PuertoV24}. It is also used as a subroutine in several geometric optimization problems, such as the watchman route problem~\cite{Mitchell13,NilssonP24}, the safari and zookeeper problems~\cite{TanH03,Bespamyatnikh03,FaiglVP13}, and even variants of convex hull~\cite{LofflerVK10,WeibelZ11,DiazBanezKPPLPS15,AntoniadisDBKS25}. The problem can also be understood as the $2$-dimensional offline version of the convex body chasing problem~\cite{FriedmanL93,Sellke20,BansalBEKU20,ArgueGGTG21,BubeckLLS23,sellke2023chasing}.

Dror, Efrat, Lubiw and Mitchell~\cite{DELM03} studied \touring under $L_2$ when the polygons are convex. In their formulation the starting point of the tour is fixed, i.e., $P_1$ degenerates to a single point. They gave an $O(nk\log n)$-time algorithm if the given convex polygons are disjoint, and an $O(nk^2\log n)$-time algorithm if they are allowed to overlap. On the other hand, they showed that the problem becomes NP-hard if the polygons are non-convex, or even if each $P_i$ consists of a pair of segments with a shared endpoint.

This raises interesting questions from the perspective of fine-grained complexity. For example, is there a \emph{truly subquadratic time} algorithm (i.e., of running time $O(n^{2-\eps})$ for some $\eps>0$) that solves \touring under $L_2$, assuming convex disjoint polygons of constant size each?

As our take to the question, we consider an interesting variant of touring \emph{orthogonal polygons} (where every edge is either horizontal or vertical) under $L_1$. This setting avoids precision issues under $L_2$ and allows us to concentrate on the combinatorial side of the problem. Here obtaining an $\Otilde(n^2)$-time algorithm is in fact an open problem. Dror~\etal~\cite{DELM03} claimed an $O(n^2)$-time algorithm but did not present a full proof. A straightforward interpretation of their sketch would lead to an algorithm in cubic time (more precisely, $O(n^2k)$ time).

\begin{question}\label{q:general}
    Is there an $\Otilde(n^2)$-time algorithm for touring a sequence of orthogonal polygons under $L_1$?
    What about $\Otilde(nk)$ or even truly-subquadratic time?
\end{question}

Guided by known results in $L_2$, we can expect that the problem becomes easier when the polygons are disjoint and/or convex. In case of orthogonal polygons the only convex polygons are axis-aligned rectangles. This leads us to the second question:
\begin{question}\label{q:special}
    Is there a faster algorithm for touring a sequence of \emph{disjoint} orthogonal polygons under $L_1$?
    Is there a faster algorithm for touring a sequence of rectangles under $L_1$?
\end{question}

\subsection{Our contribution}
\begin{table}[b]
    \centering
    \caption{Overview of our results for \touringOrthoParam}\label{table:overview}
    \begin{tabular}{lcc}
        \toprule
        input restriction & time complexity & refer to \\ \midrule
        --  & $O(n^2 \log^2 n)$ & \Cref{thm:general}\\
        --  & $\Otilde(n^{1.5}k^2)$ & \Cref{thm:dense}\\
        step-disjoint & $\Otilde(n^{2-\frac{1}{48}})$ & \Cref{thm:disjoint}\\
        step-disjoint, ortho-convex & $\Otilde(n)$ & \Cref{thm:orthoconvex}\\
        rectangles & $O(n)$  & \Cref{thm:rectangle}\\
        \bottomrule
    \end{tabular}
\end{table}

In what follows, we consider \touring for a sequence of $k$ orthogonal polygons with $n$ vertices in total and measure all distances in $L_1$ metric. We denote this problem by \touringOrthoParam. \Cref{table:overview} summarizes our results.
As a baseline, we design an $\Otilde(n^2)$-time algorithm, thereby partially answering \Cref{q:general}.

\begin{restatable}{theorem}{general}\label{thm:general}
    \touringOrthoParam is in time $O(n^2\log^2 n)$.
\end{restatable}

The proof of \Cref{thm:general} already contains some non-trivial ideas that recur in the other results. First observe that we may restrict our attention to tours whose vertices lie on the the \EMPH{grid}, i.e., the intersections of the horizontal and vertical lines through polygon vertices. Since the grid has $O(n^2)$ points, this leads to a naive dynamic program that runs in time $O(n^2 k)$. To improve the running time it is helpful to consider the case of \EMPH{step-disjoint} polygons, where any two consecutive polygons in the sequence are disjoint. (In \Cref{fig:intro}, $P_1,\dots, P_4$ is step-disjoint, while $P_1,\dots, P_5$ is not.) In this case, a shortest tour $\pi$ must enter each polygon $P_i$ at a boundary point which is also on the grid; we call these points \EMPH{portals}. Note that the number of portals on $P_i$ is at most $n$ times the number of edges in $P_i$, so the total number of portals over all polygons is $O(n^2)$. It remains to design an efficient data structure that allows one to query, for each portal on $P_{i+1}$, the shortest tour ending there. Since such tour must pass through a portal on $P_i$, we can build an additively weighted Voronoi diagrams over the portals on $P_i$ to handle these queries efficiently.

In the general case where consecutive polygons can overlap, a shortest tour may enter a polygon directly at a grid point in its interior; see for example $P_5$ in \Cref{fig:intro}. We handle the situation by ``tunneling'' through the overlapping polygons with the help of a segment-tree-based data structure.

\medskip
We do not know how to push this approach to fully address \Cref{q:general}. Nevertheless, we manage to obtain truly-subquadratic time algorithms under various natural restrictions. In the most restricted case, we consider a sequence of rectangles and answer the second half of \Cref{q:special}:

\begin{restatable}{theorem}{rectangle}\label{thm:rectangle}
    \touringOrthoParam for rectangles can be solved in $O(n) = O(k)$ time.
\end{restatable}

This highly specialized algorithm uses the convexity of rectangles in a strong way, and does not even generalize to polygons with $O(1)$ edges each (such as L-shapes). It is also challenging to solve the other extreme: a sequence of $O(1)$ polygons with $\Theta(n)$ edges each. As our main result, we propose ideas to solve both extremes for step-disjoint sequences, and to interpolate the extremes into a truly subquadratic algorithm. This answers the first half of \Cref{q:special}.

\begin{restatable}[Main theorem]{theorem}{disjoint}\label{thm:disjoint}
    \touringOrthoParam for step-disjoint orthogonal polygons can be solved in $\Otilde(n^{2-\frac{1}{48}})$ time.
\end{restatable}

The global structure of the algorithm is to split the sequence into batches of two types: a dense batch contains few polygons (but each polygon may have many edges), whereas a sparse batch contains polygons with few edges. It uses different strategies to process the two types and chain them together.

For simplicity, we expose the ideas on two extreme cases: In the dense case $k \leq n^{0.1}$ is small (thus all polygons are in one dense batch); and in the sparse case $k = \Omega(n)$ and all polygons have $O(1)$ many edges (thus all polygons are in one sparse batch).

In the dense case, we start by partitioning the plane into rectangles using a technique of De Berg and Van Kreveld~\cite{dBvK94}. The partition has the property that, for every polygon $P_i$ and rectangle $R$, the intersection $P_i \cap R$ is either a collection of horizontal stripes or a collection of vertical stripes. As a result, the tour inside $R$ is rather simple and can be reduced to one-dimensional problems. Our idea is to build a dynamic program on the grid points induced by rectangle boundaries. For each such grid point $p$ (called a \EMPH{hub}) and each index $i$ we want to compute the shortest tour visiting $P_1,\dots,P_i,p$. We need to iterate over index pairs $1\leq i \leq j \leq k$ during the computation, and the number of hubs can be bounded by $O(n^{1.5})$, so this results in an algorithm with running time $\Otilde(n^{1.5}k^2)$. Due to the quadratic dependence on $k$ this algorithm is not subquadratic for larger values of $k$, so we need a completely different strategy for the sparse case.

In the sparse case, we face a long sequence of polygons with few edges. Importantly, if the tour is visiting a pair of consecutive polygons $P_i$ and $P_{i+1}$ and makes a turn (i.e., changes from horizontal to vertical or vice versa) somewhere between their visits, then the turn can only be justified by a \emph{local} grid point, i.e., on some point that is the intersection of some horizontal and vertical lines through the vertices of~$P_i$ and~$P_{i+1}$. Since both polygons have few edges, the local grid has very small complexity. Tours making a turn between each consecutive pair can be handled using weighted Voronoi diagrams.

Unfortunately, it is possible that a tour traverses a subsequence of consecutive polygons without making any turns, and hence without snapping to the local grid. Tracking such tours efficiently is far from trivial. In order to compute tours that stay on a given horizontal line $y=y^*$, we build an auxiliary planar graph whose nodes are arranged by layers that correspond to the polygon index $1\leq i \leq k$, and inside each layer by $x$-coordinates that correspond to the $x$-coordinates where the polygons intersect the line $y=y^*$. The shortest horizontal tour on $y=y^*$ visiting the polygons corresponds to a shortest path in this graph that percolates from the first layer to the last layer. We cannot afford building these graphs from scratch for each relevant value of $y^*$, as this requires quadratic time. However, we observe that as we sweep a horizontal line bottom up, the number of updates to the planar graph is linear in total, and we can apply a multi-source distance oracle for dynamic planar graphs by Charalampopoulos and Karczmarz~\cite{CK22} to beat quadratic time.

Interpolating the dense and sparse cases requires further finesse. In particular, we need to deal with instances that contain both dense and sparse batches. A challenge is to pass information from one batch to the next and bridge the difference in strategies. The bridging is efficient only if the polygon delimiting the two batches has a small number of edges. Fortunately, we show a batching strategy such that all the delimiters have small complexity.

This concludes our overview of the main theorem. As a byproduct of the dense case we also obtain an algorithm for potentially overlapping polygons. The algorithm is truly-subquadratic (hence, faster than \Cref{thm:general}) whenever $k \leq n^{0.25-\eps}$ for any $\eps>0$.

\begin{restatable}{theorem}{dense}\label{thm:dense}
    \touringOrthoParam can be solved in $\Otilde(n^{1.5}k^2)$ time.
\end{restatable}

In another result, we consider \EMPH{ortho-convex} polygons, which are orthogonal polygons with the property that any horizontal or vertical line intersects them in a segment. If the ortho-convex polygons are step-disjoint, then a vertical edge $S$ of $P_i$ cannot have points of $P_{i-1}$ both to its left and right. Hence, for each grid point $p \in S$, the shortest tour visiting $P_1,\dots,P_{i-1},p$ can be efficiently described by the distance functions on the vertical edges of $P_{i-1}$ that are ``facing''~$S$. The distance functions in turn are piecewise linear and have slopes $-1$, $0$ or $1$. By handling and updating these functions in a persistent data structure we are able to solve the step-disjoint ortho-convex case in near-linear time, avoiding the reliance on dynamic planar graph algorithms.

\begin{restatable}{theorem}{orthoconvex}\label{thm:orthoconvex}
    \touringOrthoParam can be solved in $O(n\log n)$ time for step-disjoint ortho-convex polygons.
\end{restatable}

While this paper is about fast algorithms for \touringOrtho, the problem is also interesting from the lower bound/fine-grained complexity perspective. In particular, it does not seem to have the same quantifier structure as typical geometric problems studied in the fine-grained literature such as curve and shape similarity problems~\cite{Bringmann14,BringmannK15,BringmannN22,BringmannKKMN23,BringmannSWW24}, intersection graph problems~\cite{BringmannKKNP22,ChanCCGKLZ25}, clustering~\cite{ChanHY23}, and point-line incidence problems~\cite{GajentaanO95,BarbaCILOS19}. We raise the following natural question.

\begin{question}
    Is there a super-linear conditional lower bound for touring a sequence of convex polygons under $L_2$, or for touring orthogonal polygons under $L_1$?
\end{question}

Towards this direction we consider the generalized problem in $3$-dimensional space: touring a sequence of $k$ orthogonal polytopes of $n$ vertices in total. The problem has a straightforward $\Otilde(n^3k)$ algorithm. We show a conditional lower bound assuming the orthogonal vectors hypothesis (OVH). See \cite{Bringmann19} for an overview of popular fine-grained conjectures.

\begin{restatable}{theorem}{lowerbound}\label{thm:lowerbound}
Assuming OVH, no algorithm can solve \touringOrthoThreeParam for step-disjoint 3-dimensional polytopes in time $O(n^{2-\epsilon} \poly(k))$, for any $\eps>0$.
\end{restatable}

\paragraph*{Organization}
The rest of this paper is organized as follows. \Cref{sec:prelim} introduces the fundamental concepts and structures. \Cref{sec:general} studies the general case without input restriction. \Cref{sec:step-disjoint,sec:orthoconvex} assume step-disjointness and present truly-subquadratic time algorithms for orthogonal polygons and ortho-convex polygons, respectively. Along the way we discuss how the ideas can be used to handle (a small number of) overlapping polygons. \Cref{sec:rectangles} gives a linear time algorithm for rectangles, and \Cref{sec:lower} gives a quadratic lower bound for the three-dimensional variant of the problem.

\section{Fundamental concepts}
\label{sec:prelim}
We write $[N] := \set{1, \dots, N}$. A polygon is \EMPH{orthogonal} if all its edges are horizontal or vertical. For simplicity of presentation, we assume that all polygons are in general position, i.e., no two edges are collinear. The assumption can be removed by imposing an ordering on collinear edges. For a polygon $P$, let $x(P)$ and $y(P)$ be the sets of $x$ and $y$-coordinates of the vertices of $P$, respectively. Denote $n(P) := |x(P)| = |y(P)|$. In the \touringOrthoParam problem, we are given orthogonal polygons $P_1, \dots, P_k$ with total complexity $n := \sum_{i=1}^k n(P_i)$, and the goal is to compute (the length of) a shortest tour under $L_1$ metric that visits $P_1, \dots, P_k$. We say that the input sequence of polygons is \EMPH{step-disjoint} if $P_i$ and $P_{i+1}$ are disjoint for all $1 \leq i \leq k-1$.

For each index $i \in [k]$ and point $p$, we define $\color{BrickRed} f_i(p)$ as the length of a shortest tour that visits $P_1, \dots, P_{i-1}, p$ in sequence. We also define ${\color{BrickRed}\nxt(i,p)} := \min \set{j \geq i : p \notin P_j}$. (If $j$ does not exist then let $\nxt(i,p) := k+1$.) These definitions play a central role in our algorithms.

It is sometimes convenient to represent tours as discrete sequences instead of continuous curves. The notion of skeletons serves this purpose.
\begin{definition}
    A \EMPH{skeleton} visiting $P_1, \dots, P_k$ is a sequence $(i_1,q_1), \dots, (i_m, q_m)$ such that $1 = i_1 \leq i_2, \dots, i_m  < i_{m+1} := k+1$ and $q_t \in \bigcap_{i_t \leq i < i_{t+1}} P_i$ for all $t \in [m]$. We say that the skeleton has \emph{size} $m$ and \emph{length} $\sum_{i=2}^m \dist(q_{i-1}, q_i)$.
\end{definition}

So far as the minimum length is concerned, skeletons and tours are equivalent:
\begin{lemma}
    \label{lem:skeleton}
    The minimum length of tours visiting $P_1, \dots, P_k$ is equal to the minimum length of skeletons visiting $P_1, \dots, P_k$.
\end{lemma}

\begin{proof}
    If a tour visits $P_1, \dots, P_k$, then it visits a sequence of points $p_1 \in P_1, \dots, p_k \in P_k$, and its length is at least $\sum_{i=2}^k \dist(p_{i-1}, p_i)$. On the other hand, $(1,p_1), \dots, (k,p_k)$ is a skeleton, and its length is exactly $\sum_{i=2}^k \dist(p_{i-1}, p_i)$. This shows the $\geq$ direction.

    Conversely, if a skeleton $(i_1,q_1), \dots, (i_m,q_m)$ visits $P_1, \dots, P_k$, then we can define a tour by connecting $q_1, \dots, q_m$ in sequence. It visits $P_1, \dots, P_k$ because $i_1=1$, $i_{m+1}=k+1$, and $q_t \in \bigcap_{i_t \leq i < i_{t+1}} P_i$ for all $t \in [m]$. Its length is exactly the length of the skeleton. This shows the $\leq$ direction.
\end{proof}

The next lemma provides crucial insights into the structure of minimum length skeletons.

\begin{lemma}
    \label{lem:structure}
    Among all minimum length skeletons visiting $P_1, \dots, P_k$, there is a skeleton $(i_1,q_1), \dots, (i_m,q_m)$ with the four properties listed below. Here we denote $Q_t := \bigcap_{i_t \leq i < i_{t+1}} P_i$.
    \begin{enumerate}
        \item For each $t \in [m]$, we have $\nxt(i_t,q_t) = i_{t+1}$.
        \item For each $t \in [m]$, we have $q_t \in \bd Q_t$.
        \item For each $t \in [2,m]$, at least one of the following holds:
        \begin{enumerate}[(1)]
            \item $q_t$ is a vertex of $Q_t$ or an intersection point in $\bd Q_{t-1} \cap \bd Q_t$;
            \item $q_{t-1}, q_t$ are on horizontal edges of $Q_{t-1},Q_t$ respectively, and $x(q_{t-1}) = x(q_t)$;
            \item $q_{t-1}, q_t$ are on vertical edges of $Q_{t-1},Q_t$ respectively, and $y(q_{t-1}) = y(q_t)$;
        \end{enumerate}
        \item There exists $t \in [m]$ for which (1) holds.
    \end{enumerate}
\end{lemma}

\begin{proof}
    Among all minimum length skeletons $(i_1, q_1), \dots, (i_m,q_m)$, pick those with the minimum size $m$. Among all picked skeletons, further pick those such that $(i_2-i_1, \dots, i_{m+1}-i_m)$ is lexicographically maximum. We call these skeletons \emph{irreducible}.

    \begin{claim*}
        Every irreducible skeleton $(i_1,q_1), \dots, (i_m,q_m)$ satisfies property 1 and has $q_1 \in \bd Q_1$.
    \end{claim*}

    \begin{claimproof}
        On one hand, $\nxt(i_t, q_t) \geq i_{t+1}$ because $q_t \in Q_t = \bigcap_{i_t \leq i < i_{t+1}} P_i$ by the definition of skeletons. On the other hand, $\nxt(i_t, q_t) \leq i_{t+1}$ because otherwise we can replace $i_{t+1}$ with $\nxt(i_t,q_t)$ in the skeleton, which results in another skeleton with the same length and size but a higher lexicographic rank, a contradiction to irreducibility. This shows $\nxt(i_t,q_t) = i_{t+1}$.

        Next we show $q_1 \in \bd Q_1$. Note that $q_1 \notin P_{i_2}$ since $\nxt(1, q_1) = i_2$. In particular, $q_1 \notin Q_2$. Suppose to contradiction that $q_1 \notin \bd Q_1$, then it must be in the interior of $Q_1 \setminus Q_2$. As $q_2 \in Q_2$, we can move $q_1$ towards $q_2$ while staying inside $Q_1$. The skeleton still visits $P_1, \dots, P_k$ but the length strictly decreases, which is a contradiction.
    \end{claimproof}
    
    In the rest of the proof, we will start from an arbitrary irreducible skeleton $\sigma$ and modify it into another irreducible skeleton $\sigma'$ that is closer to satisfying properties 2 and 3. In more detail, we will \emph{move} a point $q_t$ to $q_t'$, meaning that we replace $q_t$ with $q_t'$ in $\sigma$ to obtain $\sigma'$. We will make sure that $q_t' \in Q_t$ and $\sigma'$ has the same length as $\sigma$. Once these two conditions are met, it follows that $\sigma'$ still visits $P_1, \dots, P_k$ and has the same length, size and lexicographic rank as $\sigma$, so it remains irreducible.
    
    To implement the scheme, we first iterate over $t \in [m-1]$. We have $q_t \notin P_{i_{t+1}}$ because $\nxt(i_t,q_t) = i_{t+1}$ by property 1. In particular, $q_t \notin Q_{t+1}$. Since $q_{t+1} \in Q_{t+1}$, the shortest path between $q_t, q_{t+1}$ must cross the boundary $\bd Q_{t+1}$. Therefore, we can move $q_{t+1}$ to $\bd Q_{t+1}$ while preserving the length. The resulting skeleton is still irreducible. After all iterations, we obtain an irreducible skeleton with property 2.
    
    Next we iterate over $s = 2, \dots, m$. In step $s$, the goal is to move $q_{s-1}, q_s$ such that for every $t \in [2,s]$ either (1), (2) or (3) holds. By property 2, $q_{s-1}$ and $q_s$ are on edges $e_{s-1} \subset \bd Q_{s-1}$ and $e_s \subset \bd Q_s$, respectively. Assume that $e_s$ is horizontal; the vertical case is symmetric. We distinguish four cases, illustrated in the four columns of \Cref{fig:local-repair}.
    \begin{itemize}
        \item If $e_{s-1}$ is horizontal, then we move $q_s$ along $e_s$ towards $q_{s-1}$ until it hits a vertex of $Q_s$, or until $x(q_{s-1}) = x(q_s)$. This ensures (1) or (2) for $t=s$.
        
        \item If $e_{s-1}$ is vertical and entirely below/above $e_s$, then we move $q_{s-1}$ to the top/bottom endpoint of $e_{s-1}$. Now that $q_{s-1}$ is a vertex, (1) holds for $t=s-1$; moreover, $q_{s-1}$ must lie on a horizontal edge, so we can apply the previous case to move $q_s$ and ensure (1) or (2) for $t=s$.
        
        \item If $e_{s-1}$ is vertical and entirely to the left/right of $e_s$, then we move $q_s$ to the left/right endpoint of $e_s$ and ensure (1) for $t=s$.
        
        \item If $e_{s-1}$ is vertical and intersects $e_s$, then we move $q_s$ to the intersection and ensure (1) for $t=s$.
    \end{itemize}
    In any case the goal is achieved. Furthermore, $q_{s-1}, q_s$ remain on $\bd Q_{s-1}, \bd Q_s$, respectively, and the length of the skeleton does not change. Therefore, the resulting skeleton is still irreducible and satisfies property 2. After all iterations, we obtain an irreducible skeleton that satisfies both properties 2 and 3.
    
    \begin{figure}[hbtp]
        \centering
        \includegraphics{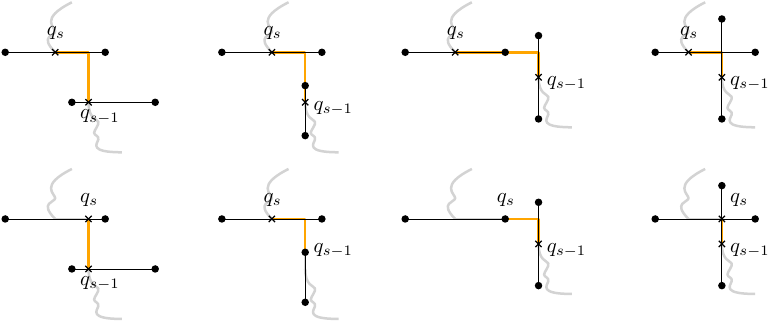}
        \caption{The four cases when $e_s$ is horizontal. The first and second rows illustrate the situation before and after the move, respectively. The orange path in the background illustrates a shortest path between $q_{s-1}$ and $q_s$.}
        \label{fig:local-repair}
    \end{figure}

    Finally, suppose that property 4 is not yet fulfilled. Then for all $t \in [2,m]$ either (2) or (3) holds. Assume by symmetry that $q_1$ is on a horizontal edge of $Q_1$, then it is not on a vertical edge as it is not a vertex. Hence $q_2$ must be on a horizontal edge of $Q_2$ and $x(q_1) = x(q_2)$. Propagating the argument, all $q_1, \dots, q_m$ are on horizontal edges and have the same $x$-coordinate. We move all points to the right by the same distance, until some point $q_t$ becomes a vertex of $Q_t$, which witnesses property 4. Clearly every point $q_s$ remains on a horizontal edge of $Q_s$ and has the same $x$-coordinate, and the length of the skeleton does not change. Therefore, irreducibility and properties 2, 3 are also preserved.
\end{proof}

Motivated by \Cref{lem:structure}, we define the \EMPH{grid} of a set of orthogonal polygons $\class{P}$ as
\[
    \grid(\class{P}) :=
    \left( \bigcup_{P \in \class{P}} x(P) \right) \times
    \left( \bigcup_{P \in \class{P}} y(P) \right) \subset \Reals^2.
\]
We define the \EMPH{portals} of $P_i$ as $\portal_i := \bd P_i \cap \grid(\{P_1, \dots, P_k\})$.

\begin{corollary}
    \label{cor:structure}
    Among all minimum length skeletons visiting $P_1, \dots, P_k$, there is a skeleton $(i_1, q_1), \dots, (i_m, q_m)$ with the four properties in \Cref{lem:structure}, as well as the property that $q_t \in \bigcup_{i_t \leq i < i_{t+1}} \portal_i$ for all $t \in [m]$.
\end{corollary}

\begin{proof}
    Take a skeleton $(i_1, q_1), \dots, (i_m, q_m)$ given by \Cref{lem:structure}. By property 4, some of $q_1, \dots, q_m$ are vertices or intersections. We split the skeleton at these points into pieces. Consider any piece but the first one. The piece starts with a point $q_t$ that is a vertex or an intersection, and no other point is a vertex or an intersection. By property 3, points in this piece either all lie on horizontal edges and have $x$-coordinate $x(q_t) \in \bigcup_{i=1}^k x(P_i)$, or all lie on vertical edges and have $y$-coordinate $y(q_t) \in \bigcup_{i=1}^k y(P_i)$. Therefore, every point $q_s$ in the piece is a portal. Moreover, since $q_s \in \bd Q_s \subseteq \bigcup_{i_s \leq i < i_{s+1}} \bd P_i$, we have $q_s \in \bigcup_{i_s \leq i < i_{s+1}} \portal_i$. For the first piece, we can make the same argument from where it ends.
\end{proof}

When the polygon sequence is step-disjoint, we get a significantly simpler structure:
\begin{corollary}
    \label{cor:structure-disjoint}
    Among all shortest tours that visit a sequence of step-disjoint polygons $P_1, \dots, P_k$, there is a tour that visits a sequence of points $p_1 \in \portal_1, \dots, p_k \in \portal_k$ such that for each $i \in [2,k]$,
    \begin{enumerate}[(1)]
        \item $p_i$ is a vertex of $P_i$; or
        \item $p_{i-1}, p_i$ are on horizontal edges of $P_{i-1}, P_i$ respectively, and $x(p_{i-1}) = x(p_i)$; or
        \item $p_{i-1}, p_i$ are on vertical edges of $P_{i-1}, P_i$ respectively, and $y(p_{i-1}) = y(p_i)$.
    \end{enumerate}
\end{corollary}

\begin{proof}
    Take a skeleton $(i_1, p_1), \dots, (i_m, p_m)$ given by \Cref{cor:structure}. Since the polygon sequence is step-disjoint, we have $\nxt(i_t,p_t) = i_t+1$ for all $t$ by definition, so inductively $i_t = t$. Hence the skeleton has form $(1,p_1), (2,p_2), \dots, (k,p_k)$, and for all $i \in [k]$ we have $Q_i = P_i$ and $p_i \in \portal_i$.
\end{proof}

\section{Touring overlapping orthogonal polygons}
\label{sec:general}
In this section, we study \touringOrtho for the general case.
Dror \etal~\cite{DELM03} stated that the problem can be solved in $O(n^2)$ time without elaborating on the proof; a straightforward interpretation of their idea actually needs $O(n^2 k)$ time. We describe here a simple algorithm in time $O(n^2\log^2 n)$, which essentially closes the gap and also serves as a baseline for our other algorithms.

\general*

\begin{proof}
    As a preparation, we construct a binary tree, where each leaf node represents an element (i.e., a singleton interval) of $[k]$, and each non-leaf node represents the union of the intervals of its two children; this is a discrete variant of a segment tree~\cite{CGAA}. We write $I(u)$ for the interval represented by node $u$ and define $Q(u) := \bigcap_{i \in I(u)} P_i$. Note that $Q(u) = Q(v) \cap Q(w)$ where $v,w$ are the children of $u$, so we can compute $Q(u)$ for all $u \in T$ in a bottom-up fashion. Then we build a point location data structure for each~$Q(u)$.
    
    It is well known that the intersection of two orthogonal polygons $P$ and $Q$ can be computed in time $O\left((n(P) + n(Q))^2 \right)$, and a point-location data structure on a polygon $P$ can be built in time $O(n(P) \log n(P))$. So the total time spent per level of $T$ is $O\big( \big(\sum_{i=1}^k n(P_i) \big)^2 \log n\big) = O(n^2 \log n)$. As there are $\lceil \log k \rceil$ levels, the construction time is $O(n^2 \log^2 n)$.
    
    \begin{figure}[htbp]
        \centering
        \includegraphics{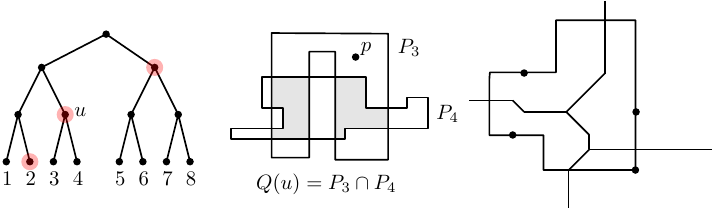}
        \caption{Left: an illustration of a segment tree over $[8]$, where the red nodes form the canonical decomposition of $[2,8]$. Middle: the gray region depicts $Q(u)$. Right: an additively weighted Voronoi diagram on four portals under the $L_1$ metric.}
        \label{fig:general}
    \end{figure}
    
    Next we describe a subroutine that, given an index $i \in [k]$ and a point $p$, computes the value $\nxt(i,p)$ in time $O(\log^2 n)$. To this end, first compute the canonical decomposition of the interval $[i,k]$; that is, find $l \leq \lceil \log k \rceil$ nodes $u_1, \dots, u_l \in T$ such that $[i,k] = I(u_1) \cup \cdots \cup I(u_l)$ is a disjoint union. Then compute the minimal $t \in [l]$ such that $p \notin Q(u_t)$. We descend from node $u_t$ towards the leaves; at each step we go to the left child $v$ if $p \notin Q(v)$, and the right child otherwise. In the end we arrive at a leaf that corresponds to exactly $\min \set{j \geq i: p \notin P_i \cap \cdots \cap P_j } = \nxt(i, p)$.
    
    Computing the canonical decomposition takes $O(\log k)$ time. Computing $t$ needs $l \leq \log k$ queries into the point-location data structures. The descent from $u_t$ to the leaf needs $O(\log k)$ queries as well. Since each query is answered in $O(\log n)$ time, the subroutine takes time $O(\log^2 n)$ as claimed.
    
    For each $i \in [k]$ and $p \in \portal_i$, we call the subroutine to compute $\nxt(i,p)$ and store the values in a look-up table. Since $\sum_{i=1}^k |\portal_i| \leq \sum_{i=1}^k 2n(P_i) \cdot n = 2n^2$, the table can be computed in time $O(n^2 \log^2 n)$.

    With these preparations, the main algorithm is a dynamic program. As the base case, we have $f_1(p)=0$ for all $p\in \portal_1$. We claim the following recursive formula:
    \begin{claim*}
        For all $2 \leq j \leq k$ and $q \in \portal_j$, we have
        \[ f_j(q) = \begin{cases}
            f_{j-1}(q) & \text{if } q \in P_{j-1}, \\
            \min \set{ f_i(p) + \dist(p,q) \::\: p \in \portal_i,\; i < j,\; \nxt(i,p) = j } & \text{otherwise}.
        \end{cases} \]
    \end{claim*}
    \begin{claimproof}
        If $q \in P_{j-1}$ then $f_j(q) = f_{j-1}(q)$ by definition. From now on we assume $q \notin P_{j-1}$. The $\leq$ direction is easy: $f_i(p)$ corresponds to a tour visiting $P_1, \dots, P_{i-1}, p \in P_i$, and $\dist(p,q)$ corresponds to a tour visiting $p, P_{i+1}, \dots, P_{j-1}, q$ since $\nxt(i,p) = j$. Their sum thus corresponds to a tour visiting $P_1, \dots, P_{j-1}, q$.
        
        For the $\geq$ direction, let $(i_1,q_1), \dots, (i_m,q_m)$ be a minimum length skeleton visiting $P_1, \dots, P_{j-1}, q$ given by \Cref{cor:structure}. Note that $q_m = q \notin P_{j-1}$ by definition, thus $m \geq 2$. Moreover, we have $i_m = j$ because $q \in \bigcap_{i_m \leq i \leq j} P_i$ but $q \notin P_{j-1}$.
        
        By the guarantee of \Cref{cor:structure}, $p := q_{m-1} \in \bigcup_{i_{m-1} \leq i < j} \portal_i$ and $\nxt(p,i_{m-1}) = i_m = j$. Hence there exists $i \in [i_{m-1}, j)$ such that $p \in \portal_i$. For this $i$ we clearly also have $\nxt(p, i) = j$. Finally, \Cref{lem:skeleton} implies that $f_j(q)$ is exactly the length of this skeleton, which is equal to $f_i(p) + \dist(p,q)$.
    \end{claimproof}
    
    The dynamic program iterates over $j=2,\dots,k$. In each iteration, we construct the additively weighted Voronoi diagram $D$ (under $L_1$ metric) on the point set
    \[ S_j := \set{ p \in \portal_i \::\: i < j,\; \nxt(i,p) = j }, \]
    where each point $p$ has weight $f_i(p)$. For each $q \in \portal_j$, if $q \in P_{j-1}$ then we assign $f_j(q) := f_{j-1}(q)$, else we assign $f_j(q) := D.\Call{dist}{q}$. The correctness follows from the claim above. After all iterations finish, we return $\min \set{ f_i(p) : p \in \portal_i, i \in [k], \nxt(p,i) = k+1}$.
    
    In each iteration $j$, constructing the Voronoi diagram takes $O(|S_j| \log n)$ time~\cite{klein1989concrete}, and the queries take $O(|\portal_j| \log n)$ time. Over all iterations, the time is $O\left( \sum_{j=1}^k (|S_j| + |\portal_j|) \log n \right)$. Note that
    \[ \sum_{j=1}^k |S_j|
    = \sum_{j=1}^k \sum_{i<j} \sum_{p \in \portal_i} \mathbf{1}_{\nxt(i,p)=j}
    = \sum_{i=1}^k \sum_{p \in \portal_i} 1
    = \sum_{i=1}^k |\portal_i|. \]
    Hence the total time is $O\left( \sum_{j=1}^k |\portal_j| \log n \right) = O(n^2 \log n)$.
\end{proof}

\section{Touring step-disjoint orthogonal polygons}
\label{sec:step-disjoint}
We move on to study step-disjoint orthogonal polygons $P_1, \dots, P_k$ and prove \Cref{thm:disjoint}, which breaks the quadratic-time barrier. As a rough outline, the algorithm decomposes the sequence of polygons into batches, where each batch either contains a small number of vertices (sparse batch), or has few polygons (dense batch). The algorithm processes the batches in order. In each round it extends the prefix shortest tours computed so far to visit one more batch. Different strategies are used for sparse and dense batches, and they are driven by different structural insights.

We will make use of two results from the literature. One is a generalized distance oracle for dynamic planar graphs \cite[Theorem 18]{CK22}:
\begin{theorem}
    \label{thm:dist-oracle}
    Let $G$ be a weighted planar digraph on $n$ nodes, with a set $U \subseteq V(G)$ of facilities. There exists a data structure maintaining $G$ under edge insertions, edge deletions and changes of $U$ in $\Otilde(n^{3/4} |U|^{1/4} + n^{4/5})$ update time. Given any $v \in V(G)$, it can compute $\min_{u \in U} \dist(u,v)$ in $\Otilde(1)$ time. The initialization time is $\Otilde(n)$.
\end{theorem}

The other is a rectangle partition of the plane with low stabbing number \cite[Lemma 3.1]{dBvK94}; we restate a slight variant and include a proof in the appendix.
\begin{lemma}
    \label{lem:rect-partition}
    In time $O(n)$, we can partition $\mathbb{R}^2$ into $n + 1$ rectangles such that
    \begin{enumerate}[(i)]
        \setlength{\itemsep}{0pt}
        \item no rectangle contains a vertex of $P_1, \dots, P_k$ in its interior;
        \item every vertical/horizontal line intersects at most $O(\sqrt{n})$ rectangles.
    \end{enumerate}
\end{lemma}

\subsection{Decomposition into batches}
Let $0 < \alpha < \beta < 1$ be constants to be determined later. (For concreteness, think of $\alpha=0.85$ and $\beta=0.98$.) We will carefully choose some polygons $P_{i_1}, P_{i_2}, \dots, P_{i_r}$ as \EMPH{delimiters}, where $1 \leq i_1 < i_2 <\dots < i_r \leq k$. They split the input sequence into \EMPH{batches} $(P_1, \dots, P_{i_1}), (P_{i_1}, \dots, P_{i_2}), \dots, (P_{i_r}, \dots, P_k)$. We say that a batch $(P_i)_{a \leq i \leq b}$ is \EMPH{sparse} if $\sum_{i=a}^b n(P_i) \leq 3n^{\beta}$; and it is \EMPH{dense} if $b-a \leq n^{1-\alpha} + 1$. Note that a batch can be both sparse and dense.

\begin{lemma}
    \label{lem:batch}
    In time $O(k)$ we can compute a set of at most $3n^{1-\beta}$ delimiters, each of complexity at most $n^{\alpha}$, such that every batch is either sparse or dense.
\end{lemma}

\begin{proof}
    We color each polygon $P_i$ blue if $n(P_i) \leq n^{\alpha}$, and red otherwise. This breaks the input sequence into blue and red blocks. We further differentiate a red block $(P_i)_{a \leq i \leq b}$ by two types: it is \emph{light red} if $\sum_{i=a}^b n(P_i) \leq n^{\beta}$, and \emph{dark red} otherwise. Clearly there are at most $n^{1-\beta}$ dark red blocks.
    
    Now we fix two consecutive dark red blocks, and consider the sequence of polygons $(P_i)_{a \leq i \leq b}$ in between. It must start with a blue block, alternate between light red and blue blocks, and end with a blue block. We mark $P_a$ and $P_b$ as delimiters. We then iterate over $i = a, \dots, b$ and keep a running sum $N$. At iteration $i$ we let $N := N + n(P_i)$. If $N > n^{\beta}$ \emph{and} $P_i$ is blue, then we mark $P_i$ as a delimiter and reset $N$ to zero.

    We run the procedure between every pair of consecutive dark red blocks, and return all the marked delimiters. The running time is clearly $O(k)$. Next we argue that the required properties hold.

    The number of delimiters is bounded as follows. Between every pair of consecutive dark red blocks, we mark (i) a delimiter in the beginning; (ii) a delimiter in the end; and (iii) one or more delimiters in the middle. Over all consecutive pairs, the total contribution of (i)(ii) is at most twice the number of dark red blocks, that is $2n^{1-\beta}$. The total contribution of (iii) is at most $n^{1-\beta}$ because each time we mark a delimiter this way, the running sum must have exceeded $n^\beta$.

    By construction every delimiter is blue (i.e., has complexity at most $n^\alpha$). Now consider an arbitrary batch $(P_i)_{a \leq i \leq b}$. There are only two cases:
    \begin{itemize}
        \item It consists of a dark red block and two adjacent blue polygons. We have $b - a \leq n^{1-\alpha} + 1$ since each red polygon has complexity at least $n^\alpha$ and the total complexity is at most $n$. Hence the batch is dense.
        
        \item It is formed when we scan the sequence between consecutive dark red blocks. The scanning procedure ends the batch as soon as the running sum $N$ exceeds $n^{\beta}$ and we see a blue polygon. Since each light red block and each blue polygon can contribute at most $n^{\beta}$ to the sum, we have $\sum_{i=a}^b n(P_i) \leq 3n^{\beta}$. Hence the batch is sparse.\qedhere
    \end{itemize}
\end{proof}

From now on we assume that the decomposition in \Cref{lem:batch} is computed. Our algorithm will process the batches in order and pass necessary information from one batch to the next. When we finish processing a batch $(P_i)_{a \leq i \leq b}$, we would have computed $f_b(p)$ for all $p \in \portal_b$. We employ different strategies to process sparse and dense batches. Let us explain the high-level ideas.

Inside a sparse batch, we define \emph{terminals} as the local analogue of portals. More precisely, these are the boundary points on the grid induced by the batch. Clearly there are at most $(3n^{\beta})^2 = 9n^{2\beta}$ terminals. We show that there exists a global shortest tour that traverses this batch in one of three ways: (i) it visits a terminal on each polygon in the batch; (ii) it travels horizontally throughout the batch; or (iii) it travels vertically throughout the batch. To handle (i), we use a simple dynamic program over terminals, which costs only polylogarithmic time per terminal. To handle (ii) or (iii), we build a dynamic planar graph and use a line sweep procedure to extract shortest tours to all portals on the last polygon.

Inside a dense batch, the number of terminals can be quadratic, so (i) can no longer be handled efficiently even though the structural insight remains valid. To overcome this barrier, we partition the plane into rectangles with a low stabbing number, which effectively summarizes the geometry. This way, we obtain only $O(n^{3/2})$ (instead of $O(n^2)$) points worth considering in the dynamic program, which we call \emph{hubs}. Since the number of polygons is small in a dense batch, this already implies a subquadratic-time dynamic program. However, a hub might not lie on the boundary of any polygon, so we have to deal with the intricate problem of converting $f$-values on portals to $f$-values on hubs, and vice versa. We apply two line sweeps based on dynamic planar graphs, one before and one after the dynamic program, to convert between the two worlds.

\subsection{Processing sparse batches}
\begin{theorem}
    \label{thm:sparse-batch}
    Let $(P_i)_{a \leq i \leq b}$ be a sparse batch. Given $f_a(p)$ for all $p \in \portal_a$, we can compute $f_b(p)$ for all $p \in \portal_b$ in time $\Otilde(n^{2\beta} + n^{1+\alpha} + n^{(7+\alpha)/4} + n^{9/5})$.
\end{theorem}

We devote the section to proving \Cref{thm:sparse-batch}. Throughout we fix a sparse batch $(P_i)_{a \leq i \leq b}$. We define the \EMPH{terminals} of $P_i$ as $\terminal_i := \bd P_i \cap  \grid(\set{P_i}_{a \leq i \leq b})$. Note that $\terminal_i \subseteq \portal_i$. Since a sparse batch has complexity $O(n^{\beta})$, the number of terminals is $O(n^{2\beta})$.

\begin{lemma}
    \label{lem:structure-sparse}
	Let $p \in \portal_b$. Among all shortest tours visiting $P_1, \dots, P_{b-1}, p$, there is a tour that satisfies one of the following properties:
    \begin{enumerate}
        \item it visits $\portal_1, \dots, \portal_{a-1}, \terminal_a, \dots, \terminal_{b-1}, p$;
        \item it visits $\portal_1, \dots, \portal_a$ and stays on a horizontal line afterwards until it reaches $p$;
        \item it visits $\portal_1, \dots, \portal_a$ and stays on a vertical line afterwards until it reaches $p$.
    \end{enumerate}
\end{lemma}

\begin{proof}
    Let $\pi$ be a shortest tour visiting $P_1, \dots, P_{b-1}, p$ given by \Cref{cor:structure-disjoint}. We claim that $\pi$ is the tour we are looking for.
    
    To this end, define $p_1 \in \portal_1, \dots, p_{b-1} \in \portal_{b-1}$ as in \Cref{cor:structure-disjoint}. First suppose that none of $p_a, \dots, p_{b-1}$ is a vertex. If $p_{b-1}$ is on a horizontal edge (thus not on a vertical edge), then \Cref{cor:structure-disjoint} guarantees that $x(p_{b-2}) = x(p_{b-1}) = x(p)$ and $p_{b-2}$ is on a horizontal edge as well. Propagating the argument in reverse order of time, we conclude that $x(p_a) = \cdots = x(p_{b-1}) = x(p)$, so $\pi$ satisfies property~3. In the symmetric case that $p_b$ is on a vertical edge, we conclude that $\pi$ satisfies property~2.

    Next suppose that some of $p_a, \dots, p_{b-1}$ are vertices. We split the sequence $p_a, \dots, p_{b-1}$ at these vertices into subsequences. Consider any subsequence but the first one. It must start from a vertex, say $p_i$, and does not contain any other vertex. So either all points in it share the same $x$-coordinate $x(p_i)$, or all share the same $y$-coordinate $y(p_i)$. Hence all are terminals. For the first subsequence, we can make the same argument from where it ends. We have thus shown that $p_i \in \terminal_i$ for all $a \leq i < b$, hence $\pi$ satisfies property~1.
\end{proof}

For index $a \leq i \leq b$ and point $p \in \portal_i$, let $f_i^1(p)$ be the minimum length among all tours that visit $\portal_1, \dots, \portal_{a-1}, \terminal_a, \dots, \terminal_{i-1}, p$. Let $f_i^2(p)$ be the minimum length among all tours that visit $\portal_1, \dots, \portal_a, P_{a+1}, \dots, P_{i-1}, p$ and stay horizontal after $\portal_a$. Let $f_i^3(p)$ be defined similarly to $f_i^2(p)$ but now the tour must stay vertical after $\portal_a$. \Cref{lem:structure-sparse} states that $f_i(p) = \min\{f_i^1(p), f_i^2(p), f_i^3(p)\}$, so it remains to compute the three values separately.

\paragraph{Type-1 tours}
For $a \leq i \leq b$ and $p \in \portal_i$, we have the recursive formula
\[
    f_i^1(p) =
    \begin{cases}
        f_i(p) & i=a, \\
        \min \set{ f_{i-1}^1(q) + \dist(p,q) : q \in \terminal_{i-1} } & i \in (a,b].
    \end{cases}
\]
This leads to a straightforward dynamic program: For $i = a+1, \dots, b$, build a Voronoi diagram $D$ over all $q \in \terminal_{i-1}$ using $f_{i-1}^1(q)$ as additive weights, then let $f_i^1(p) := D.\Call{dist}{p}$ for all $p \in \terminal_i$. (In the last iteration we do so for all $p \in \portal_b$.)

Let us analyze the running time of iteration $i$. Building the Voronoi diagram takes time $\Otilde(\card{\terminal_{i-1}})$. Querying the Voronoi diagram takes time $\Otilde(\card{\terminal_i})$ in total. (For the last iteration querying takes time $\Otilde(\card{\portal_b})$, which is $\Otilde(n^{1+\alpha})$ because the delimiter $P_b$ has complexity $n^\alpha$.) Recall that $\sum_{i=a}^b \card{\terminal_i} =O(n^{2 \beta})$ by sparsity, the time complexity over all iterations is $\Otilde(n^{2\beta} + n^{1+\alpha})$.

\paragraph{Type-2/3 tours}
These two types are symmetric, so we focus on type 2. We sort $\class{Y} := y(P_1) \cup \cdots \cup y(P_k)$ in increasing order and cut it into contiguous groups of size $n^{1-\alpha}$ each.

Fix an arbitrary group $Y \subset \class{Y}$. We sweep a horizontal line $y \in Y$ bottom up. At each sweep step we aim to compute $f_2(p)$ for all $p \in \portal_b \cap y$. To this end we build a weighted planar digraph $G^{(y)}$ as follows (see \Cref{fig:line-sweep-1} for an illustration).
\begin{itemize}
    \item The node set consists of multiple layers $S \cup X_a \cup \cdots \cup X_b$, where
    \begin{align*}
        X_i &:= \begin{cases}
            x(P_i) & i = a \\
            x(P_i) \cup x(P_{i-1}) & a < i \leq b \\
        \end{cases} \\
        S &:= \{ s_{x,h} \;:\; x \in X_a\;, h \in Y \}.
    \end{align*}
    The nodes in each layer $X_i$ are arranged on a line in the natural order. A node $x \in X_i$ is \emph{active} if $(x,y)$ is a portal.

    \item For every pair of neighboring nodes $x, x' \in X_i$ in the same layer, we add an edge $x \leftrightarrow x'$ of weight $|x-x'|$.

    \item For every active node $x \in X_i$ and its copy $x' \in X_{i+1}$, we add an edge $x \to x'$ of weight 0.

    \item For every $s_{x,h} \in S$, we add an edge $s_{x,h} \to x$. The weight is $f_a(x,h)$ if $(x,h)$ is a portal, and $\infty$ otherwise.

    \item Finally, we specify $U^{(y)} := \{s_{x,y} : x \in X_a \}$ as the facilities.
\end{itemize}

\begin{figure}[htbp]
    \centering
    \includegraphics{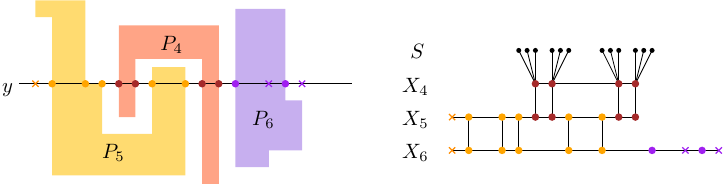}
    \caption{The structure of the plane graph $G^{(y)}$ for a batch of three polygons $P_4, P_5, P_6$. Active nodes are drawn as dots, and inactive nodes are drawn as crosses. Edges are oriented downwards, whose weights are omitted for clarity.}
    \label{fig:line-sweep-1}
\end{figure}

It is clear from the construction that $G^{(y)}$ is planar. Note that $|S| = O(n^\alpha \cdot |Y|) = O(n)$, and $\sum_{i=a}^{b} |X_i| \leq 2\sum_{i=a}^{b} n(P_i) = O(n)$ by sparsity. So the number of nodes is bounded by $O(n)$. Moreover, we have $f_2(x,y) = \min_{u \in U^{(y)}} \dist(u,x)$ for all $x \in X_b$. In other words, the $f^2$-values can be retrieved by querying the corresponding nodes in the graph.

Observe that $G^{(y)}$ changes marginally when we sweep $y$ for one step: The set of nodes do not change, and at most two inter-layer edges are inserted/removed due to activation/deactivation of nodes. The facilities $U^{(y)}$ change completely, but the size is bounded by $|x(P_b)| \leq n^\alpha$.

Applying the data structure from \Cref{thm:dist-oracle}, we get the following running time:
\begin{itemize}
    \setlength{\itemsep}{0pt}
    \item Initialization takes time $\Otilde(n)$;
    \item In each sweep step, it takes time $\Otilde(n^{(3+\alpha)/4} + n^{4/5})$ to update the graph, and time $\Otilde(n^\alpha)$ to query the distances to portals.
\end{itemize}

Finally, we sum over all groups $Y$. There are $n^\alpha$ groups, thus the same number of initializations. Over all groups there are $O(n)$ sweep steps. So the total time of handling type-2 tours in this batch is
\[
    \Otilde\left( n^{1+\alpha} + n^{1+(3+\alpha)/4} + n^{1+4/5} + n^{1+\alpha} \right)
    = \Otilde\left( n^{1+\alpha} + n^{(7+\alpha)/4} + n^{9/5} \right).
\]

\paragraph{Summary} Having computed $f_b^1, f_b^2, f_b^3$, we can compute $f_b(p) = \min\{f_b^1(p), f_b^2(p), f_b^3(p)\}$ for all $p \in \portal_b$. Correctness is guaranteed by \Cref{lem:structure-sparse}. This completes the proof of \Cref{thm:sparse-batch}.

\subsection{Processing dense batches}
\begin{theorem}
    \label{thm:dense-batch}
    Let $(P_i)_{a \leq i \leq b}$ be a dense batch. Given $f_a(p)$ for all $p \in \portal_a$, we can compute $f_b(p)$ for all $p \in \portal_b$ in time $\Otilde(n^\gamma)$, where
    $\gamma := \max \{ \frac{9}{5}, 1+\alpha, \frac{7+\alpha}{4}, \frac{5}{2}-\alpha, \frac{7}{2}-2\alpha, \frac{9-2\alpha}{4}, \frac{33-2\alpha}{16} \}$.
\end{theorem}

We devote the section to proving \Cref{thm:dense-batch}. Throughout we fix a dense batch $(P_i)_{a \leq i \leq b}$, which by definition satisfies $b - a = O(n^{1-\alpha})$. We also fix a rectangle partition from \Cref{lem:rect-partition}. Property (i) in \Cref{lem:rect-partition} implies that for any polygon $P_i$ and rectangle $R$, the intersection $P_i \cap R$ is either a collection of vertical stripes, a collection of horizontal stripes, or empty.\footnote{Technically, $P_i$ and $R$ might intersect only on the boundary of $R$. Nevertheless, this corner case does not affect our arguments. Our algorithm still works by considering the adjacent rectangle $R'$ for which $P_i \cap R'$ is indeed either a collection of vertical stripes or a collection of horizontal stripes.}
If it is non-empty, we say that $P_i$ is \EMPH{vertical} or \EMPH{horizontal} in $R$, respectively.

A \EMPH{hub} is a point $p \in \grid(\set{P_1, \dots, P_k})$ that appears on the boundary of some rectangle. Equivalently, every hub is on the intersection of some vertical/horizontal line through the polygon vertices and some rectangle boundary. Since every line may generate at most $O(\sqrt{n})$ hubs by property (ii) in \Cref{lem:rect-partition}, and since there are $2n$ lines in total, we have the following:

\begin{observation}
    \label{obs:n-hub}
    The number of hubs is $O (n^{3/2})$.
\end{observation}

Before we continue, we introduce a canonical decomposition of a tour into phases. Assume that tour $\pi$ visits $p_a, \dots, p_b$ in sequence, where $p_i$ is the first point in $P_i$ visited by $\pi$. Let $q_1, \dots, q_h$ be the sequence of hubs visited by $\pi$. A subtour $\pi[q_a, q_{a+1}]$ between consecutive hubs is called a \EMPH{progression phase} if it contains some point $p_i$ where $a \leq i \leq b$ (generally it may contain a sequence of points $p_i, \dots, p_j$). The subtours sandwiched between progression phases are called \EMPH{teleportation phases}. See \Cref{fig:progress-teleport} for an example.

\begin{figure}[htb]
    \centering
    \includegraphics{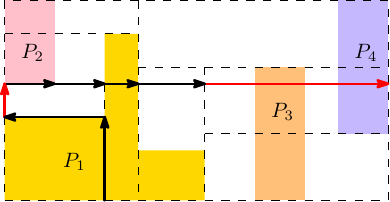}
    \caption{An illustration of a tour that visits four polygons $P_1, P_2, P_3, P_4$ and nine hubs. It has two progression phases (in red) and two teleportation phases (in black).}
    \label{fig:progress-teleport}
\end{figure}

\begin{lemma}
    \label{lem:structure-dense}
    Among all shortest tours that visit $P_1, \dots, P_b$ in sequence, there is a tour such that every progression phase is confined within a rectangle of the decomposition.
\end{lemma}

\begin{proof}
    We take a shortest tour $\pi$ that conforms to \Cref{cor:structure-disjoint}. Reccall that $\pi$ snaps to $\grid(\set{P_1, \dots, P_k})$, its intersection points with the rectangle boundaries are hubs. Recall that every progression phase of $\pi$ is a subtour between consecutive hubs, so it must be confined within a rectangle.
\end{proof}

Zooming into a progression phase, we show that vertical and horizontal movements are in a sense independent.
\begin{lemma}
    \label{lem:v-h-independence}
    Let $R$ be a rectangle from the decomposition. Consider a shortest tour and its subtour $\pi$ inside $R$. Suppose that $\pi$ starts at hub $p$, visits $P_i, \dots, P_j$ in sequence, and ends at hub $q$. Let $I_\mathrm{h} \subseteq \{i, \dots, j\}$ be the indices corresponding to horizontal polygons in $R$, and $I_\mathrm{v} \subseteq \{i, \dots, j\}$ be the indices corresponding to vertical polygons in $R$. Then the vertical movements of $\pi$, all combined, is a shortest vertical tour that visits $x(p), (P_t)_{t \in I_\mathrm{h}}, x(q)$ in $R$. Similarly, the horizontal movements of $\pi$, all combined, is a shortest horizontal tour that visits $y(p), (P_t)_{t \in I_\mathrm{v}}, y(q)$ in $R$.
\end{lemma}
\begin{proof}
    We prove the statement for vertical movements; the proof for horizontal movements is symmetric.
    Assume towards contradiction that there is a shorter vertical tour $\eta$ that visits $x(p), (P_t)_{t \in I_\mathrm{h}}, x(q)$ in $R$. We replace the vertical movement of $\pi$ that reaches a polygon $P_i$ by the movement of $\eta$ that reaches $P_i$, without changing the order of the movements. The resulting tour starts at hub $p$, visits $P_i, \dots, P_j$ in sequence, ends at hub $q$ and is shorter, contradicting the optimality of $\pi$.
\end{proof}

Our algorithm mirrors the alternation between progression and teleportation phases: In a progression phase we walk inside a rectangle and visit a sequence of new polygon(s), while in a teleportation phase we may jump far in space but do not visit any new polygon. The algorithm is outlined as follows:
\begin{itemize}
    \item We use the known $f_a(\cdot)$ values on portals of $P_a$ to initialize $f_a(\cdot)$ on hubs.

    \item We run a dynamic program that iterates over $j = a+1, \dots, b$. In each iteration we compute $f_j(\cdot)$ by implementing one alternation of progression and teleportation phases.
    
    \item In the end, all hubs receive the correct $f_b(\cdot)$ values. We project them back to $f_b(\cdot)$ values on portals of $P_b$, thereby achieving the goal of the batch.
\end{itemize}
We will now elaborate on the algorithm.

\paragraph{Initialization}
For initialization, we want to derive $f_a(\cdot)$ values on hubs from the already computed $f_a(\cdot)$ values on $\portal_a$. We use the same line-sweep procedure as we did for type-2/3 tours in sparse batches. Namely, we sweep a global line $y$ and maintain a dynamic planar graph $G^{(y)}$. The graph is defined as before, so the time analysis for updating edges/facilities remains valid.

The queries are slightly different though: For every hub $p = (x,y)$ and every index $a \leq i \leq b$, we initialize $f_i(p) := \min\{\Call{query}{x'} + |x'-x|, \Call{query}{x''} + |x''-x|\}$, where $(x',y)$ is the closest portal on $P_i$ to the left of $p$, and $(x'',y)$ is the closest portal on $P_i$ to the right of $p$. Hence, for every sweep step we make $O((b - a) \sqrt{n}) = O(n^{3/2-\alpha})$ queries. Summing over $O(n)$ sweep steps, the total query time is $\Otilde(n^{5/2-\alpha})$.

Altogether, the initialization costs time $\Otilde(n^{1+\alpha} + n^{(7+\alpha)/4} + n^{5/2-\alpha} + n^{9/5})$.

\paragraph{Dynamic program}
For each $a \leq j \leq b$, rectangle $R$ and hub $q \in \bd R$, we define $g_j(R,q)$ as the minimum length among all tours that visit $P_1, \dots, P_j, q$ and whose final leg is a progression phase in $R$. The next two lemmas relate $f$ and $g$.

\begin{lemma}
    \label{lem:teleportation}
    Fix $a \leq j \leq b$. Given $g_j(R,p)$ for all rectangles $R$ and hubs $p \in \bd R$, we can compute $f_j(q)$ for all hubs $q$ in time $\Otilde(n^{3/2})$.
\end{lemma}

\begin{proof}
    Recall that $f_j(q)$ is the minimum length of a tour that visits $P_1, \dots, P_j, q$. Note that
    \[ f_j(q) = \min_{\text{hub }p} \left( \dist(p,q) + \min_{R \ni p} g_j(R,p) \right). \]
    The $\leq$ direction is clear since the right hand side always composes a tour that visits $P_1, \dots, P_j, q$ in sequence. For the $\geq$ direction, consider a tour defining $f_j(q)$. If its last leg was a progression phase in some $R$, then by definition $f_j(q) = g_j(R,q)$. Else, its last leg was a teleportation phase from some hub $p$, prior to which was a progression phase in some $R$. So $f_j(q) = g_j(R,p) + \dist(p,q)$.
    
    To compute it efficiently, we build a Voronoi diagram $D$ on hubs, with each hub $p$ additively weighted by $\min_{R \ni p} g_j(R,p)$. Then for each hub $q$, we compute $f_j(q) := D.\Call{dist}{q}$. Correctness follows directly from the recursion above, so it remains to analyze time complexity. The Voronoi diagram has $O(n^{3/2})$ sites (hubs) by \Cref{obs:n-hub}. Computing the additive weight of each site takes time $O(1)$, so the construction of the diagram takes time $\Otilde(n^{3/2})$. Every hub gives rise to one query into the diagram, which takes time $\Otilde(1)$. So the total query time amounts to $\Otilde(n^{3/2})$.
\end{proof}

\begin{lemma}
    \label{lem:progression}
    Fix $R$ and $a \leq j \leq b$. Given $f_a(p), \dots, f_{j-1}(p)$ for all hubs $p \in \bd R$, we can compute $g_j(R,q)$ for all hubs $q \in \bd R$ in time $O (n^{1-\alpha} N \log N)$, where $N$ is the number of hubs on $\bd R$.
\end{lemma}

\begin{proof}
    Write $i \preceq j$ if $a \leq i \leq j \leq b$ and each of $P_i, \dots, P_j$ intersects $R$. In this case, we can partition $\set{i, \dots, j} =: I_\mathrm{v} \cup I_\mathrm{h}$ based on whether a polygon is vertical or horizontal in $R$. Let $\text{hcost}(R, i, j, x, x')$ be the length of the shortest horizontal tour that visits $x, (P_t)_{t \in I_\mathrm{v}}, x'$ in $R$; let $\text{vcost}(R, i, j, y, y')$ be the length of the shortest vertical tour that visits $y, (P_t)_{t \in I_\mathrm{h}}, y'$ in $R$. We claim that
    \[
        g_j(R,q) = \min_{\substack{\text{hub }p \in \bd R\\i \preceq j}} \big(
            f_{i-1}(p) + \text{hcost}(R, i, j, x(p), x(q)) + \text{vcost}(R, i, j, y(p), y(q))
        \big).
    \]
    The $\leq$ direction is clear.
    For the $\geq$ direction, we consider a shortest tour that defines $g_j(R,q)$. Its final leg is a progression phase in $R$ that visits a non-empty sequence of polygons $P_i, \dots, P_j$ where $a \leq i \leq b$ and $i \preceq j$. This part of the tour may contain horizontal and vertical movements, and by \Cref{lem:v-h-independence} the two movements have length $\text{hcost}(R, i, j, x(p), x(q))$ and $\text{vcost}(R, i, j, y(p), y(q))$, respectively. Before this, the tour must have visited $P_1, \dots, P_{i-1}$, so it has length $f_{i-1}(p)$ by definition.

    In what follows, we show how to compute the recursion efficiently. For each edge $e \subset R$, we define
    \[
        g_j^e(R,q) := \min_{\substack{\text{hub }p \in e\\i \preceq j}} \big(
            f_{i-1}(p) + \text{hcost}(R, i, j, x(p), x(q)) + \text{vcost}(R, i, j, y(p), y(q))
        \big).
    \]
    Clearly $g_j = \min \set{g_j^e : e \text{ is an edge of } R}$. So it remains to compute $g_j^e$ for each $e$.
    
    By symmetry we focus on the case that $e$ is the left edge. Since $x(p) = x(e)$ is a fixed value for all hubs $p \in e$, we can factor out the horizontal cost and write
    \[
        g_j^e(R,q) = \min_{i \preceq j} \left(
        \text{hcost}(R, i, j, x(e), x(q)) + \min_{\text{hub } p \in e} (f_{i-1}(p) + \text{vcost}(R, i, j, y(p), y(q)))
        \right).
    \]
    
    From now on we fix $i \preceq j$. We compute $\text{hcost}(R, i, j, x(e), x(q))$ for all $q$ by running Dijkstra's algorithm on a layered graph. Specifically, let $L \subseteq \set{i, \dots, j}$ be the indices of the vertical polygons in $R$. For each $l \in L$ we define a layer $V_l := x(P_l \cap R)$. We also introduce a source layer $V_0 := \set{x(e)}$ as well as a sink layer $V_\infty$ that contains the $x$-coordinates of hubs on $R$. The edges are defined as follows:
    \begin{itemize}
        \setlength{\itemsep}{0pt}
        \item We arrange the nodes in each layer on a real line, and add an edge between every pair of adjacent nodes. The edge weight is the difference of their $x$-coordinates.

        \item For $l \in \set{0} \cup L$ and every $x \in V_l$, we find the smallest index $l' > l$ such that $P_{l'} \cap R$ does \emph{not} cover $x$. (We set $l' := \infty$ if no such index exists.) Then we add an edge from $x$ to
        \begin{itemize}
            \setlength{\itemsep}{0pt}
            \item the closest node $x' \in V_{l'}$ to the left of $x$; and
            \item the closest node $x' \in V_{l'}$ to the right of $x$.
        \end{itemize}
        The edge weights are both $|x - x'|$.
    \end{itemize}
    By construction, a shortest tour from $x(e)$ to $x \in V_\infty$ in the graph corresponds to a shortest horizontal tour that visits $x(e), (P_l)_{l \in L}, x$ in $R$. So to compute $\text{hcost}(R, i, j, x(e), x(q))$ for all hubs $q$, it suffices to run Dijkstra's algorithm from $x(e)$ and read off the distances to all $x \in V_\infty$. Note that the number of nodes is bounded by $O(N)$, and the number of edges is linear in the number of nodes. Hence Dijkstra's algorithm runs in time $O(N \log N)$.

    In a similar fashion we compute $\min_{p \in e} (f_{i-1}(p) + \text{vcost}(R, i, j, y(p), y(q)))$ for all $q$. This time, let $L \subseteq \set{i, \dots, j}$ collect the indices of the horizontal polygons in $R$. For each $l \in L$ we define a layer $V_l := y(P_l \cap R)$. We introduce a source node $s$ and two more layers $V_0, V_\infty$. The layer $V_0$ contains the $y$-coordinates of hubs on $e$. The layer $V_\infty$ contains the $y$-coordinates of all hubs on $R$. The edges are defined as follows:
    \begin{itemize}
        \setlength{\itemsep}{0pt}
        \item We add an edge from the source $s$ to every node $y \in V_0$ with weight $f_{i-1}(x(e),y)$.
        
        \item We arrange the nodes in each layer on a real line, and add an edge between every pair of adjacent nodes. The edge weight is the difference of their $y$-coordinates.

        \item For $l \in \set{0} \cup L$ and every $y \in V_l$, we find the smallest index $l' > l$ such that $P_{l'} \cap R$ does \emph{not} cover $y$. (We set $l' := \infty$ if no such index exists.) Then we add an edge from $y$ to
        \begin{itemize}
            \setlength{\itemsep}{0pt}
            \item the closest node $y' \in V_{l'}$ above $y$; and
            \item the closest node $y' \in V_{l'}$ below $y$.
        \end{itemize}
        The edge weights are both $|y - y'|$.
    \end{itemize}
    Then we run Dijkstra's algorithm from $s$ and read off the distances to all $y \in V_\infty$. The running time is again $O(N \log N)$.
    
    After we have computed the two terms, we add them for each $q$. We iterate the procedure for every $i \preceq j$ and thereby obtain $g_e$, from which we compute $g$. Since there are at most $b-a = O(n^{1-\alpha})$ iterations, the total running time is as claimed.
\end{proof}

\begin{corollary}
    \label{cor:progression}
    Fix $a \leq j \leq b$. Given $f_a(p), \dots, f_{j-1}(p)$ for all hubs $p \in \bd R$, we can compute $g_j(R,q)$ for all rectangles $R$ and hubs $q \in \bd R$ in time $\Otilde(n^{5/2-\alpha})$.
\end{corollary}

\begin{proof}
    We apply \Cref{cor:progression} for each rectangle $R$. Since each hub appears on at most four rectangles, the sum over all $N$'s is bounded by four times the number of hubs, which is $O(n^{3/2})$.
\end{proof}

Now we can assemble the dynamic program. For $j = a+1, \dots, b$, we apply \Cref{cor:progression} to compute $g(\cdot,\cdot,j)$ from $f_a(\cdot), \dots, f_{j-1}(\cdot)$, then we apply \Cref{lem:teleportation} to compute $f_j(\cdot)$ from $g(\cdot,\cdot,j)$. Hence, each iteration corresponds to an alternation of progression and teleportation phases. When the last iteration finishes, we obtain $f_b(q)$ for all hubs $q$.

Since the number of iterations is $b-a = O(n^{1-\alpha})$, the time complexity is bounded by $\Otilde(n^{1-\alpha} \cdot (n^{3/2} + n^{5/2-\alpha})) = \Otilde(n^{7/2-2\alpha})$.

\paragraph{Projection}
As the final step, we project the computed $f_b$-values on hubs to $f_b$-values on portals. Not surprisingly, we apply yet another line sweep, but the details are quite different from initialization.

Let $\class{Y} := \bigcup_{i=1}^k y(P_i)$. We sort $\class{Y}$ in increasing order and break it into contiguous groups of size $n^{(2\alpha-1)/4}$ each. Now fix an arbitrary group $Y \subset \class{Y}$. We sweep a horizontal line $y \in Y$ bottom up. Since each $y$ hits $O(\sqrt{n})$ hubs, there are at most $O(n^{(2\alpha+1)/4})$ relevant hubs for the group. We collect the $x$-coordinates of these hubs into a set $X$.

For each line $y$, we build a weighted planar digraph $G^{(y)}$ as follows:
\begin{itemize}
    \item The node set consists of multiple layers $S \cup X_a \cup \cdots \cup X_b$, where
    \begin{align*}
        X_i &:= \begin{cases}
            X \cup x(P_i) & i = a \\
            X \cup x(P_i) \cup x(P_{i-1}) & a < i \leq b \\
        \end{cases} \\
        S &:= \{ s_{i,x,h} \;:\; a \leq i \leq b,\; x \in X\;, h \in Y \}.
    \end{align*}
    The nodes in each layer $X_i$ are arranged on a line in the natural order. A node $x \in X_i$ is \emph{active} if $(x,y)$ is a portal.

    \item For every pair of neighboring nodes $x, x' \in X_i$ in the same layer, we add an edge $x \leftrightarrow x'$ of weight $|x-x'|$.

    \item For every active node $x \in X_i$ and its copy $x' \in X_{i+1}$, we add an edge $x \to x'$ of weight 0.

    \item For every $s_{i,x,h} \in S$, we add an edge $s_{i,x,h} \to x$. The weight is $f_i(x,h)$ if $(x,h)$ is a hub, and $\infty$ otherwise.

    \item Finally, we specify $U^{(y)} := \{s_{i,x,y} : a \leq i \leq b, x \in X \}$ as the facilities.
\end{itemize}
See \Cref{fig:line-sweep-2} for an illustration. Clearly the graph is planar. Since the complexity of this batch is at most $O(n)$, and $b - a = O(n^{1-\alpha})$, we can bound $|G^{(y)}| = O(n + n^{1-\alpha} \cdot |X| \cdot |Y|) = O(n)$ and $|U^{(y)}| = O(n^{1-\alpha} \cdot|X|) = O(n^{(5-2\alpha)/4})$.

\begin{figure}[htbp]
    \centering
    \includegraphics{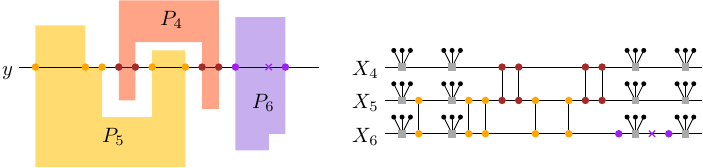}
    \caption{The structure of the plane graph $G^{(y)}$ for a batch of three polygons $P_4, P_5, P_6$. Active nodes are drawn as dots, and inactive nodes are drawn as crosses. The nodes that correspond to $X$ are drawn as squares. Edges are oriented downwards, whose weights are omitted for clarity.}
    \label{fig:line-sweep-2}
\end{figure}

When we sweep $y$ for one step, the graph $G^{(y)}$ changes only marginally. The set of nodes do not change, and at most two edges are inserted/removed due to activation/deactivation of nodes. The facilities $U^{(y)}$ change completely. To retrieve $f_b(p)$ for every $p=(x,y) \in \portal_b \cap y$, we simply query the distance from $U^{(y)}$ to $x \in X_b$. The data structure in \Cref{thm:dist-oracle} takes $\Otilde(n)$ time at initialization. For each sweep step, it takes time $\Otilde(n^{(17-2\alpha)/16})$ to update the graph and facilities, and time $\Otilde(n^{\alpha})$ to query the distance to portals.

Finally, we sum over all groups $Y$. There are $n/|Y| = n^{(5-2\alpha)/4}$ groups, thus the same number of initializations. Over all groups there are $O(n)$ sweep steps. So the total time is
\[
    \Otilde\left( n^{1+(5-2\alpha)/4} + n^{1+(17-2\alpha)/16} + n^{1+\alpha} \right)
    = \Otilde\left( n^{(9-2\alpha)/4} + n^{(33-2\alpha)/16} + n^{1+\alpha} \right).
\]

\paragraph{Summary}
\Cref{thm:dense-batch} now follows by chaining the initialization, dynamic program and projection.

\subsection{Wrapping up the proofs of \Cref{thm:disjoint,thm:dense}}
Having shown our batching strategy and the way both sparse and dense batches are handled, we are now ready to prove our main result.

\disjoint*

\begin{proof}
    We first compute a decomposition into sparse and dense batches by \Cref{lem:batch}, which takes time $O(k)$.
    If the first batch is sparse we can set $f_1(p)=0$ for all portals $p \in \portal_1$. Because the batch is sparse there are at most $O(n^{1+\beta})$ of portals on $P_1$. If the first batch is dense we alter its initialization procedure slightly: We use a single source node instead of a source layer in the graph $G^{(y)}$, and all edges incident to the source have weight zero; during the line sweep we do not update the facilities. The time bound remains valid. Then we chain \Cref{thm:sparse-batch,thm:dense-batch} to process the batches in order. Each batch requires time $\Otilde(n^\gamma)$ where
    \[
        \gamma := \max \left\{
            \frac{9}{5},\; 2\beta,\; 1+\alpha,\;
            \frac{7+\alpha}{4},\; \frac{5}{2}-\alpha,\; \frac{7}{2}-2\alpha,\;
            \frac{9-2\alpha}{4},\; \frac{33-2\alpha}{16}
        \right\}.
    \]
    If the last batch is sparse we return the minimum of $f_k(p)$ over all $p \in \portal_k$. This gives the shortest tour length by \Cref{cor:structure-disjoint}. Since the batch is sparse, the time is bounded by $O(|\portal_k|) = O(n^{1+\beta})$. If the last batch is dense we alter the projection procedure slightly. Instead of making a layer with portals in the graph $G^{(y)}$, we create a sink node and add an edge from each node in the last layer to the sink. The minimum tour length is exactly the distance from the source layer to the sink. The time bound remains valid. If we are interested in reconstructing the actual tour, we can use standard backtracking in the dynamic program.
    
    We choose parameters $\alpha := 5/6 < 47/48 =: \beta$, thus $\gamma = 47/24$. Since there are at most $3n^{1-\beta} = O(n^{1/48})$ batches, the entire algorithm runs in time $\Otilde(n^{95/48})= \Otilde(n^{2-\frac{1}{48}})$.
\end{proof}

Our handling of the dense cases yields a proof of \Cref{thm:dense} as a byproduct.
\dense*
\begin{proof}
    Consider the entire instance as a dense batch and use the adaptations for the first and last batch as in the proof of \Cref{thm:disjoint}. The runtime follows from the fact that the batch has length $k$ instead of $O(n^{1-\alpha})$.
\end{proof}

\section{Touring step-disjoint ortho-convex polygons}
\label{sec:orthoconvex}
In this section we study step-disjoint \EMPH{ortho-convex} polygons. A polygon $P$ is ortho-convex if $P \cap l$ is a (possibly empty) line segment for every horizontal and vertical line $l$. We show the following theorem.

\orthoconvex*

Our algorithm iterates over $P_1, \dots, P_k$. In iteration $i \in [k]$, the goal is to compute for every edge $e \subset \bd P_i$ the restriction $f_i^e$ of $f_i$ on $e$. After all iterations, we simply output the minimum of $f_k^e(p)$ over all edges $e \subset \bd P_k$ and portals $p \in e$. The correctness is guaranteed by \Cref{cor:structure-disjoint}.

Here is the challenge: In the worst case the total complexity of $f_i$ over all $i \in [k]$ is $\Theta(n^2)$, so we cannot explicitly compute them all. We bypass the issue by representing them in some (interrelated) data structures which avoid explicit evaluation on too many points.

In \Cref{sec:mountain-range}, we formalize the interface and guarantee of the data structure. Then in \Cref{sec:update}, we explain how to implement the algorithm efficiently by applying the data structure as a blackbox. Finally, in \Cref{sec:ds} we implement the data structure itself.

\subsection{Dynamic mountain ranges}
\label{sec:mountain-range}
We call a continuous, piecewise linear function $f: [a,b] \to \Reals$ a \EMPH{mountain range} if every piece has slope $0$, $-1$ or $1$. Its \emph{complexity} $|f|$ is the number of pieces it contains. A \emph{dynamic} mountain range $f$ is a mountain range that is initially zero everywhere and undergoes the following updates:
\begin{itemize}
    \item $\Call{restrict}{a,b}$: restrict the domain to $[a,b]$.
    
    \item $\Call{shift}{\delta}$: add $\delta \in \Reals$ to $f$ pointwise.
    \item $\Call{relax}{\lambda, \gamma}$: given $\lambda \in \set{-1,1}$ and $\gamma \in \Reals$, replace $f$ by $f'$ where
    \[ f'(x) := \min \set{ f(x),\; \lambda x + \gamma } \]
    
    \item $\Call{join}{g}$: denote the domain of $f$ by $[a,b]$. Given another mountain range $g$ over domain $[b,c]$ such that $f(b) = g(b)$, replace $f$ by $f'$ where
    \[ f'(x) := \begin{cases}
        f(x) & x \in [a,b],\\
        g(x) & x \in [b,c].
    \end{cases} \]
\end{itemize}
It also supports the query operation $\Call{evaluate}{x}$: given any $x$ in the domain, report $f(x)$.

\begin{lemma}
    \label{lem:mountain-range}
    There is a fully-persistent data structure that maintains a dynamic mountain range $f$ in amortized time $O(\log |f|)$ per operation.
\end{lemma}

We defer the proof to \Cref{sec:ds}. As a remark, \emph{persistence} is needed in our application because we sometimes want to make parallel updates to the same mountain range and keep all the results.

\subsection{The algorithm}
\label{sec:update}

The notion of (dynamic) mountain ranges is motivated by the following consideration. \Cref{cor:structure-disjoint} implies that for each $i \in [k]$ the function $f_i^e$ is a mountain range whose pieces meet at portals; moreover, $\sum_{e \subset \bd P_i} |f_i^e| = O(n)$ due to ortho-convexity. Note that in the worst case $\sum_{i=1}^k \sum_{e \subset P_i} |f_i^e|$ can be $\Theta(n^2)$, so we cannot afford to represent all of them explicitly. The trick is to break down these long mountain ranges into fragments that are interrelated by the shift/relax operations. In this way, we only need to represent a small subset of fragments by the dynamic mountain range data structure and infer the others when requested.

More specifically, our algorithm initializes a dynamic mountain range $f_1^e$ for every edge $e \subset \bd P_1$, using \Cref{lem:mountain-range}. Now focus on iteration $i \in [2,k]$. We inductively assume that the previous iteration has computed $f_{i-1}^e$ for all edges $e \subset \bd P_{i-1}$. We call $p \in \bd P_{i-1} \cup \bd P_i$ a \EMPH{terminal} if $p \in \grid(\set{P_{i-1}, P_i})$. Note that a terminal is a portal, but not vice versa. We split every edge of $P_{i-1}, P_i$ at terminals, resulting in \EMPH{fragments}. There are at most $O(n(P_{i-1}) + n(P_i))$ terminals (and thus fragments) because each vertical/horizontal line hits the boundary of an ortho-convex polygon at most twice. It is straightforward to obtain $f_{i-1}^T$ for all fragments $T \subset \bd P_{i-1}$ using \Call{restrict}{}. Based on these dynamic mountain ranges, the next two lemmas allow us to compute $f_i^S$ for all fragments $S \subset \bd P_i$ efficiently. Finally, we can piece together $f_i^e$ for all edges $e \subset \bd P_i$ using \Call{join}{}, and finish the goal of iteration $i$.

\begin{lemma}
    \label{lem:dist-on-terminal}
    Given dynamic mountain ranges $f_{i-1}^T$ for all fragments $T \subset \bd P_{i-1}$, we can compute $f_i(p)$ for all terminals $p \in \bd P_i$ in time $O((n(P_{i-1})+n(P_i))\log n)$.
\end{lemma}
\begin{proof}
    We start with a structural observation. Let $p$ be a terminal lying on the boundary of $P_i$. Let $q$ be the point in $P_{i-1}$ closest to $p$ with $x(p) = x(q)$ (if it exists), and let $q'$ be the point in $P_{i-1}$ closest to $p$ with $y(p) = y(q')$ (if it exists). Note that they are uniquely defined since $P_{i-1}$ is ortho-convex. Let $V_{i-1}$ be the set of vertices in $P_{i-1}$. We claim that $f_i(p) = \min \set{ f_{i-1}(r) + \dist(r,p) : r \in V_{i-1} \cup \set{q,q'}}$.
    
    Indeed, note that $q,q' \in \bd P_{i-1}$ due to step-disjointness. Since $p$ is a portal and $q,q'$ share a coordinate with $p$, both points are portals. Now consider a shortest tour that ends at $p$. By \Cref{cor:structure-disjoint}, we may assume that the tour visits either a vertex of $P_{i-1}$ or $q$ or $q'$. On the other hand, the concatenation of a tour realizing $f_{i-1}(r)$ and the tour $r \to p$ visits $P_1,\ldots,P_i$ in order for any point $r\in V_{i-1}\cup \{q,q'\}$. Therefore, the claim holds.

    As preprocessing, we compute $f_{i-1}(r)$ for all $r \in V_{i-1}$ using \Call{evaluate}{}, and then construct the additively weighted Voronoi diagram of $V_{i-1}$ where each point $r$ has weight $f_{i-1}(r)$. This takes time $O(n(P_{i-1})\log n)$.
    
    After the preprocessing, we are ready to compute $f_i(p)$ for each terminal $p \in \bd P_i$. To this end, we query the Voronoi diagram to obtain $\min \{ f_{i-1}(r)+\dist(r,p) : r\in V_{i-1}\}$. We compute the two terminals $q,q'$ (as defined before) and $f_{i-1}(q),f_{i-1}(q')$ using \Call{evaluate}{}. Finally, we can compute $f_i(p)$ using the claim above.
    
    Since the computation takes $O(\log n)$ time for each terminal by \Cref{lem:mountain-range}, and there are $O(n(P_{i-1}) + n(P_i))$ terminals on $\bd P_i$, the overall running time is $O((n(P_{i-1})+n(P_i))\log n)$.
\end{proof}

\begin{lemma}
    \label{lem:dist-on-fragment}
    Given dynamic mountain ranges $f_{i-1}^T$ for all fragments $T \subset \bd P_{i-1}$, together with $f_i(p)$ for all terminals $p \in \bd P_i$, we can compute $f_i^S$ for any fragment $S \subset \bd P_i$ in time $O(\log n)$.
\end{lemma}
\begin{proof}
    We show how to do this for a horizontal fragment $S$ of $P_i$; the vertical fragments can be handled symmetrically.
    Let terminals $q,q'$ be the left and right ends of $S$. Let $T$ be the horizontal fragment in $P_{i-1}$ closest to $S$ such that $x(S) = x(T)$. For any point $p \in S$, we claim that
    \[
        f_i(p) = \min \begin{Bmatrix}
            f_{i-1}(\proj_T(p))+|y(S)-y(T)|,\\
            f_{i-1}(q) + x(p)-x(q),\\
            f_{i-1}(q') + x(q')-x(p)
        \end{Bmatrix}.
    \]
    To see this, observe that there are three possibilities of the last visited point $r \in P_{i-1}$ to reach $p$ by \Cref{cor:structure-disjoint}.
    For the first case that $x(p) = x(r)$, since $P_{i-1}$ is ortho-convex, all points in $P_{i-1}$ with $x$-coordinate $x(p)$ must appear on the same side of $S$. Hence $r \in T$, or in other words $r = \proj_T(p)$.
    For the second case that $y(p) = y(r)$, since $P_{i-1}, P_i$ are disjoint, $x(r) \notin [x(q), x(q')]$. In particular, the tour $r \to p$ goes through $q$ or $q'$.
    For the third case that $r$ is a vertex, $x(r) \notin (x(q), x(q'))$ by definition of fragments, thus the tour $r \to p$ goes through $q$ or $q'$. Therefore, the claim holds.

    To compute $f_i^S$, we first compute the closest horizontal fragment $T \subset \bd P_{i-1}$ with $x(S)=x(T)$. The first term in the claim can be obtained by $f_{i-1}^T.\Call{shift}{|y(S)-y(T)|}$. The second term is a linear function in $x(p)$ with slope $+1$, namely $x(p) \mapsto x(p) + (f_{i-1}(q)-x(q))$. Similarly, the third term is a linear function in $x(q)$ with slope $-1$, namely $x(p) \mapsto -x(p) + (f_{i-1}(q')+x(q'))$. So the minimum can be computed by calling $\Call{relax}{1, f_{i-1}(q)-x(q)}$ and then $\Call{relax}{-1, x(q')+f_{i-1}(q')}$. All the operations cost $O(\log n)$ time by \Cref{lem:mountain-range}.
\end{proof}

Chaining \Cref{lem:dist-on-terminal,lem:dist-on-fragment}, the description of iteration $i$ is complete. Since each iteration takes $O((n(P_{i-1})+n(P_i))\log n)$ time, overall we need $O\left(\sum_{i=1}^k n(P_i) \log n\right) = O(n\log n)$ time. In the end, we return $\min_{p\in\portal_k} f_k(p)$, which is correct by \Cref{cor:structure-disjoint}. Since $|\portal_k| \leq 2n$ by ortho-convexity, we can find the minimum in $O(n\log n)$ time. This concludes the proof of \Cref{thm:orthoconvex}.

\subsection{Implementing a dynamic mountain range}
\label{sec:ds}
We will now prove \Cref{lem:mountain-range}.
\begin{proof}
    We will first design a non-persistent data structure based on splay trees. We treat the mountain range as a polyline whose vertices are stored in a splay tree. Specifically, every node $u$ in the tree stores two numbers $x_u \in [a,b]$ and $\delta_u \in  \Reals$, and the tree is ordered by the first number. We write $y_u := \sum_v \delta_v$ where $v$ runs over all ancestors of $u$ (including itself). The idea is to interpret $(x_u, y_u)$ as a vertex of the polyline. While $y_u$ is not explicitly stored, we make sure that it is invariant under splay operations. For example, when the tree performs a zig step as in \Cref{fig:splay-zig}, we set $\delta_u' := \delta_u + \delta_v$, $\delta_v' := -\delta_u$, $\delta_w' := \delta_w + \delta_u$ and keep everything else in place.

    \begin{figure}[h]
        \centering
        \includegraphics{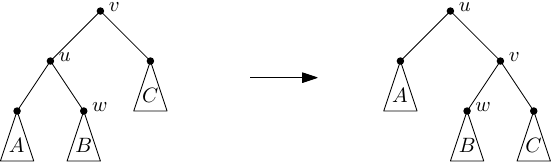}
        \caption{A zig step of the splay tree.}
        \label{fig:splay-zig}
    \end{figure}

    Under this interpretation, we can later speak of \emph{inserting a vertex $(x,y)$ to the tree}. What we really mean is that we create a new node $u$ with $x_u := x$, insert it to the splay tree, and set up the $\delta_v$'s so that $y_u = y$ and the remaining $y_v$'s do not change. The time cost is dominated by the cost of inserting an element to the tree.
    
    Later we also invoke a well-known fact: Given any decreasing/increasing property $\mathcal{Z}(x)$ in terms of $x$, we can binary search for the greatest/smallest $x_u$ in the tree such that $\mathcal{Z}(x_u)$ holds. The time cost is asymptotically equal to the cost of finding a given key in the tree.
    
    Initially, the tree contains only two nodes that correspond to vertices $(-\infty,0)$ and $(\infty,0)$. The subsequent queries and updates are handled as follows.
    \begin{itemize}
        \item For query $\Call{evaluate}{x}$, we search for the greatest $x_u \leq x$ and the smallest $x_v \geq x$ in the tree. Note that $(x_u, y_u)$ and $(x_v, y_v)$ together bound the line segment which contains $x$, so we just output $\frac{y_v-y_u}{x_v-x_u} (x - x_u) + y_u$.
        
        \item For $\Call{restrict}{a, b}$, we insert new vertices $(a, \Call{evaluate}{a})$ and $(b, \Call{evaluate}{b})$ into the tree, then we remove all nodes $u$ in the tree such that $x_u \notin [a,b]$.
        
        \item For $\Call{shift}{\delta}$, we increase $\delta_u$ by $\delta$ where $u$ is the root of the splay tree.

        \item For $\Call{relax}{\lambda, \gamma}$, if $\lambda = 1$ then we search for the greatest $x_u$ such that $y_u \geq x_u + \gamma$. This works because the slopes in $f$ is bounded by $1$, meaning that the property $f(x) \geq x + \gamma$ is decreasing. Let node $v$ be the successor of $u$. Note that the segment between $(x_u,y_u)$ and $(x_v,y_v)$ has slope $0$ or $-1$, and it intersects the line $y = x+\gamma$ at some point $(x^*, y^*)$. We compute $(x^*, y^*)$ and insert it as a new vertex into the tree, and remove all nodes $w$ such that $x_w < x^*$. The $\lambda = -1$ case is symmetric.
        
        \item For $\Call{join}{g}$, we join the two trees that represent $f$ and $g$.
    \end{itemize}

    The correctness of the algorithm is self-explanatory. Since each update boils down to a constant number of elementary operations in the splay tree, the amortized time per update is $\log(|f|)$.

    Finally, notice that our data structure is pointer-based. It is known that any such data structure can be transformed into a fully-persistent one with only a constant-factor overhead in time complexity~\cite{DRISCOLL198986}. This finishes the proof.
\end{proof}

\section{Touring rectangles}
\label{sec:rectangles}
This section aims to prove the following theorem:

\rectangle*

Let \touringInt be the problem of computing a shortest tour that visits a given sequence of $k$ intervals in $\Reals$. In the first lemma below, we reduce \touringOrtho for rectangles to two instances of \touringInt. In the second lemma below, we solve \touringInt in linear time. The two lemmas together imply \Cref{thm:rectangle}. We remark that the approach generalizes to higher dimensional boxes.

\begin{lemma}
    If \touringInt is in time $O(k)$, then \touringOrtho for rectangles is in time $O(k)$.
\end{lemma}
\begin{proof}
    Let $R_1, \dots, R_k$ be the input sequence of rectangles, where $R_i = X_i \times Y_i$ for $i \in [k]$. We compute the shortest one-dimensional tour $\tau$ (respectively $\tau'$) for the \touringInt instance $X_1, \dots, X_k$ (respectively $Y_1, \dots, Y_k$). See \Cref{fig:rectangles} for an example. We claim that $\norm{\tau} + \norm{\tau'}$ is equal to $l^*$, the length of the shortest tour that visits $R_1, \dots, R_k$.

    \begin{figure}[hb]
    	\centering
    	\includegraphics[page=7]{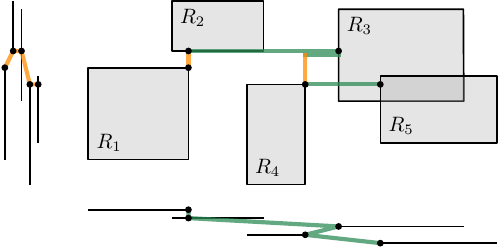}
    	\caption{A shortest tour that visits a sequence of rectangles. The tour decomposes into horizontal movements (green) and vertical movements (orange), which are the shortest tours of the projected intervals on the $x$- and $y$-axis, respectively.}
    	\label{fig:rectangles}
    \end{figure}
    
    Consider any shortest tour $\pi$ that visits $R_1, \dots, R_k$. Its projection $\pi_x$ to the $x$-axis must visit the interval sequence $X_1, \dots, X_k$ (thus a candidate for $\tau$), and its projection $\pi_y$ to the $y$-axis must visit the interval sequence $Y_1, \dots, Y_k$ (thus a candidate for $\tau'$). Hence $\norm{\tau} + \norm{\tau'} \leq \norm{\pi_x} + \norm{\pi_y} = \norm{\pi} = l^*$.

    Conversely, assume that $\tau$ (respectively $\tau'$) goes through $x_1 \in X_1, \dots, x_k \in X_k$ (respectively $y_1 \in Y_1, \dots, y_k \in Y_k$) in sequence. Then $(x_i,y_i) \in R_i$ for all $i \in [k]$, and this sequence of points specifies a tour that visits $R_1, \dots, R_k$ with length exactly $\norm{\tau} + \norm{\tau'}$. Hence $\norm{\tau} + \norm{\tau'} \geq l^*$.
\end{proof}

\begin{lemma}
     \touringInt can be solved in time $O(k)$.
\end{lemma}
\begin{proof}
    Let $X_1, \dots, X_k \subset \Reals$ be the input intervals. We first assume $X_1 \cap X_2 = \emptyset$ and $X_{k-1} \cap X_k = \emptyset$. In this special case, we solve the problem by a simple greedy algorithm. Let $x_1$ be the point in $X_1$ closest to $X_2$. For $i = 2, \dots k$, let $x_i$ be the point in $X_i$ closest to $x_{i-1}$. We return the tour $\tau = (x_1,\dots,x_k)$. Clearly the algorithm runs in time $O(k)$ and returns a tour $\tau$ visiting $X_1, \dots, X_k$. It remains to prove that $\tau$ has minimum length.
    
    We say that index $i$ is \emph{crucial} if it is an endpoint or a turning point, i.e.,
    \begin{itemize}
        \item $i \in \{1,k\}$; or
        \item for the smallest index $j \geq i$ such that $x_j \neq x_i$, we have $x_{i-1},x_j < x_i$ or $x_{i-1},x_j > x_i$.
    \end{itemize}
    Let $1 = i_1 < i_2 < \cdots < i_l = k$ be the subsequence of crucial indices. For all $t \in [l]$, we claim that $x_{i_t}$ is an end of $X_{i_t}$; moreover, it is the left end if $x_{i_t-1} < x_{i_t}$ and the right end otherwise.
    
    For $t = 1$, the point $x_{i_1} = x_1$ is an end of $X_1$ because the algorithm picks it as the closest point to $X_2$, and $X_1 \cap X_2 = \emptyset$. Similarly, for $t = l$, the point $x_{i_l} = x_k$ is an end of $X_k$ because the algorithm picks $x_k$ as the closest point to $x_{k-1} \in X_{k-1}$, and $X_{k-1} \cap X_k = \emptyset$. For $1 < t < l$, we have $x_{i_t-1} \neq x_{i_t}$ by cruciality. In particular, $x_{i_t-1} \notin X_{i_t}$ and $x_{i_t}$ must be an end of $X_{i_t}$ because the algorithm picks $x_{i_t}$ as the point in $X_{i_t}$ closest to $x_{i_t-1}$. Moreover, if $x_{i_t-1} < x_{i_t}$ then the point must be the left end, and otherwise it must be the right end. The claim is now proven.

    It is easy to see that $x_{i_1}, \dots, x_{i_l}$ alternate as left/right ends of the corresponding intervals. In particular, $\dist(x_{i_t}, x_{i_{t+1}}) = \dist(X_{i_t}, X_{i_{t+1}}) > 0$ for all $t \in [l-1]$. So any tour that visits $X_1, \dots, X_k$ in sequence has length at least
    \[
        \sum_{t=1}^{l-1} \dist(X_{i_t}, X_{i_{t+1}})
        = \sum_{t=1}^{l-1} \dist(x_{i_t},x_{i_{t+1}})
        = \norm{\tau}.
    \]
    This shows the optimality of $\tau$.

    Finally, we reduce the general case to the special case above. Given any instance $X_1, \dots, X_k$, if the intervals have a common intersection then the optimal tour clearly has length zero. Otherwise, we compute the largest index $l$ such that $\bigcap_{i=1}^l X_i\neq \emptyset$, as well as the smallest index $r$ such that $\bigcap_{i=r}^k X_i\neq \emptyset$. Note that $l < r$. Construct a new sequence $\bigcap_{i=1}^l X_i, X_{l+1}, \dots, X_{r-1}, \bigcap_{i=r}^k X_i$. By the choices of $l,r$, the first two intervals in the new sequence are disjoint, so are the last two intervals. We invoke the greedy algorithm on the new sequence and output the tour it returns.
    
    Clearly the algorithm returns in time $O(k)$ a tour visiting $X_1, \dots, X_k$. To see correctness, we claim that every tour visiting $X_1, \dots, X_k$ must visit the new sequence. Denoting $[a,b] := \bigcap_{i=1}^l X_i$, the two ends $a,b$ come from the left end of some interval $X_i$ and the right end of some interval $X_j$, respectively, where $i,j \in [l]$. Every tour visiting $X_1, \dots, X_l$ must visit a point $x \geq a$ (in order to visit $X_i$) and a point $x' \leq b$ (in order to visit $X_j$). Regardless of the order of visiting $x,x'$, the tour must visit the interval $[a,b] = \bigcap_{i=1}^l X_i$ by continuity. By a symmetric argument, every tour visiting $X_r, \dots X_k$ must visit $\bigcap_{i=r}^k X_i$. Therefore, every tour visiting $X_1, \dots, X_k$ must visit the new sequence, as claimed.
\end{proof}

\section{Touring orthogonal polytopes}
\label{sec:lower}

In this section, we study the generalized problem \touringOrthoThreeParam. \Cref{lem:structure} naturally generalizes to three dimensions, but for our purpose we will only state a weaker lemma. We define $\grid := \left( \bigcup_{i=1}^k x(P_i) \right) \times \left( \bigcup_{i=1}^k y(P_i) \right) \times \left( \bigcup_{i=1}^k z(P_i) \right) \subset \Reals^3$.
\begin{lemma}
    \label{lem:grid-3d}
    Among all shortest tours visiting polytopes $P_1, \dots, P_k \subset \Reals^3$, there is a tour visiting $p_1, \dots, p_k$ where $p_i\in P_i \cap \grid$ for all $i\in [k]$.
\end{lemma}

\begin{proof}
    Let $\mathcal H$ be the set of axis-aligned planes that pass through the grid. We start from an arbitrary shortest tour $\pi$ and iteratively snap it to $\mathcal H$. Let $\gamma \subseteq \pi$ be a maximal subtour that lies on an axis-aligned plane $h(\gamma) \notin \mathcal H$. Let $e,e'$ be segments in $\pi$ preceding and succeeding $\gamma$; both are perpendicular to $h(\gamma)$ by maximality of $\gamma$. We translate $\gamma$ in a direction perpendicular to the plane, towards the further endpoint of $e$, until $h(\gamma)$ hits that endpoint or $h(\gamma) \in \mathcal{H}$. The translation either decreases the number of segments in $\pi$, or makes $h(\gamma) \in \mathcal{H}$. We repeat until there is no $\gamma$ satisfying the initial condition. Upon termination, every segment in the $x$-direction in $\pi$ has the form $[x_1,x_2] \times \set{y} \times \set{z}$ where $y \in \bigcup_{i=1}^k y(P_i)$ and $z \in \bigcup_{i=1}^k z(P_i)$. Segments in the $y$-direction and $z$-direction satisfy the symmetric properties. These properties imply that $\pi$ visits $p_1, \dots, p_k$ where $p_i \in P_i \cap \grid$ for all $i \in [k]$.
\end{proof}

Lemma~\ref{lem:grid-3d} yields a straightforward algorithm for \touringOrthoThreeParam:

\begin{lemma}
    \label{lem:3dim}
    \touringOrthoThreeParam can be solved in $\Otilde(n^3 k)$ time.
\end{lemma}

\begin{proof}
By \Cref{lem:grid-3d}, we can restrict our attention to tours whose vertices lie on the grid. Note that the grid has size $O(n^3)$. We connect each grid point to its closest neighbors along each axis direction, with the edge weight being the geometric distance. We take $k$ copies of this grid and call them \emph{layers}. For every layer $i$, we add a zero-weight directed edge from each grid point in $P_i$ to the corresponding point in layer $i+1$. We attach a source vertex $s$ with outgoing zero-weight edges to all grid points in $P_1$. Symmetrically, we attach a sink vertex $t$ with incoming zero-weight edges from all grid points in $P_k$. In the constructed directed weighted graph, the shortest path from $s$ to $t$ is a shortest tour visiting the input sequence. It can be computed via Dijkstra's algorithm in $\Otilde(n^3 k)$ time.
\end{proof}

On the other hand, we prove a quadratic lower bound by a reduction from the \OV problem: Given two collections of $d$-dimensional binary vectors $A=\{a_1,\dots, a_N\} \subset \{0,1\}^d$ and $B=\{b_1,\dots, b_N\} \subset \{0,1\}^d$, decide whether there exist $i, j \in [N]$ such that the inner product of $a_i$ and $b_j$ is zero. The orthogonal vectors hypothesis (OVH) conjectures that \OV is not in time $O(N^{2-\varepsilon}\poly(d))$ for any $\eps>0$ \cite{ovh}.

\lowerbound*

\begin{proof}
Given an \OV instance $A=\{a_1,\dots, a_N\}\subset \{0,1\}^d$ and $B=\{b_1,\dots, b_N\}\subset\{0,1\}^d$, we take $c := 10d$, $D = 20dN$, and construct a sequence of $k := 4d+2$ orthogonal polytopes each of complexity $O(N)$. (More precisely, we construct a sequence of objects each consisting of axis-aligned rectangles and segments. They can be inflated into orthogonal polytopes while keeping the same complexity up to a constant factor.) Throughout the construction we refer to the $x$- and $y$-axes as horizontal and vertical directions, respectively.

Each polytope has a comb-like structure; see Figure~\ref{fig:ovh}. The teeth of the comb are either all horizontal or all vertical, separated by distance at least $c-1$; perpendicular to the teeth is a base rectangle. In more detail, each vertical comb uses the base rectangle $[c,cN+1]\times \{D\} \times [0,5]$ with vertical line segments attached. Each horizontal comb uses the base rectangle $\{D\} \times [c,cN+1] \times [0,5]$ with horizontal line segments attached.

\begin{figure}[htb]
    \centering
    \includegraphics{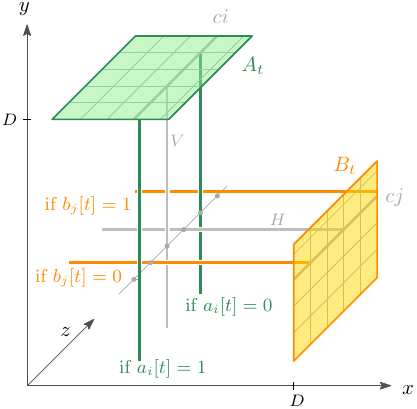}\hspace*{0.5cm}
    \includegraphics{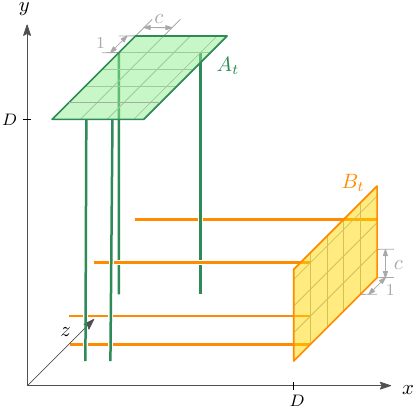}
    \caption{Left: the $i$-th segments in $A_t$ and $V$, together with the $j$-th segments in $B_t$ and $H$. Right: the comb $A_t$ that corresponds to $(a_1[t], a_2[t], a_3[t], a_4[t]) = (0,1,1,0)$, and the comb $B_t$ that corresponds to $(b_1[t], b_2[t], b_3[t], b_4[t]) = (0,0,1,0)$.}
    \label{fig:ovh}
\end{figure}

For each dimension $1\leq t\leq d$, we construct a vertical comb $A_t$ and a horizontal comb $B_t$. The comb $A_t$ has $N$ vertical line segments, where the $i$-th segment encodes the coordinate $a_i[t]$ and is defined as $\set{ci + a_i[t]} \times [c,D] \times \set{4-4a_i[t]}$. Similarly, the comb $B_t$ has $N$ vertical segments, where the $j$-th segment encodes the coordinate $b_j[t]$ and is defined as $[c,D] \times \set{cj+b_j[t]} \times \set{1+4b_j[t]}$.

We also define an auxiliary vertical comb $V$ and an auxiliary horizontal comb $H$. The comb $V$ has $N$ vertical segments where the $i$-th segment is defined as $\set{ci} \times [c,D] \times \set{3}$. The comb $H$ has $N$ horizontal segments where the $j$-th segment is defined as $[c,D] \times \set{cj} \times \set{2}$.

Finally, we define the touring sequence as
\[\mathcal S:= (V,~H,A_1,B_1,V,~H,A_2,B_2,V,~\dots, H,A_d,B_d,V,~H).\]
The sequence has length $k=4d+2$ and every comb has complexity $O(N)$. The sequence is also step-disjoint: It alternates between vertical and horizontal combs, whose bases clearly do not intersect, and whose line segments have different $z$-coordinates and thus do not intersect either.

\begin{claim*}
    Fix $i, j \in \set{1,\dots,N}$ and $t \in \set{1, \dots, d}$. The shortest (local) tour that starts from $(ci,cj,2)$, visits $H, A_t, B_t, V$ in sequence and ends at $(ci,cj,3)$, has length 13 if $a_i[t] = b_j[t] = 1$ and 7 otherwise.
\end{claim*}

\begin{claimproof}
    Any tour that starts from $(ci,cj,2)$, visits $H, A_t, B_t, V$ in sequence and ends at $(ci,cj,3)$, must have form
    \[
        (ci,cj,2)
        \to \left(ci+a_i[t],\, y,\, 4-4a_i[t]\right)
        \to \left(x,\, cj+b_j[t],\, 1+4b_j[t]\right)
        \to (ci,cj,3).
    \]
    Clearly the minimum length is attained when $y = cj$ and $x = ci+a_i[t]$. In that case, the tour becomes
    \begin{align*}
        (ci,cj,2)
        &\to \left(ci+a_i[t],\, cj,\, 4-4a_i[t]\right) \\
        &\to \left(ci+a_i[t],\, cj+b_j[t],\, 1+4b_j[t]\right) \\
        &\to (ci,cj,3).
    \end{align*}
    The first step has length $a_i[t] + |4a_i[t]-2| = a_i[t] + 2$, the second has length $b_j[t] + |4a_i[t] + 4b_j[t] - 3|$, and the third has length $a_i[t] + b_j[t] + |4b_j[t]-2| = a_i[t] + b_j[t] + 2$. The total length is thus
    \[ 2a_i[t] + 2b_i[t] + |4a_i[t] + 4b_j[t] - 3| + 4. \]
    This is equal to $13$ when $a_i[t] = b_j[t] = 1$, and $7$ otherwise.
\end{claimproof}

\begin{claim*}
    The shortest tour visiting $\mathcal S$ has length at most $8d+1$ if and only if $(A,B)$ contains a pair of orthogonal vectors.
\end{claim*}

\begin{claimproof}
For the ``if'' direction, assume that $(A,B)$ contains a pair of orthogonal vectors $a_i, b_j$. In other words, $(a_i[t], b_j[t]) \neq (1,1)$ for all $1 \leq t \leq d$. We apply the claim above for every $t$, chain the resulting length-7 tours together, and add unit length paths to bridge the gaps between $V,H$. This gives a tour of length $8d+1$ visiting $\mathcal S$ in sequence.

For the ``only if'' direction, assume there exists a tour $\pi=(p_1,p_2,\dots,p_k)$ of length at most $8d+1$ visiting $\mathcal S$ in sequence. By Lemma~\ref{lem:grid-3d}, we may assume for all $l$ that $p_l \in \grid(\mathcal{S})$; that is,
\begin{align*}
    x(p_l) &\in \set{D} \cup \set{ci+\delta : i \in \set{1,\dots,N}, \delta\in \set{0,1}}, \\
    y(p_l) &\in \set{D} \cup \set{cj+\delta : j \in \set{1,\dots,N}, \delta\in \set{0,1}}, \\
    z(p_l) &\in \set{0,\dots,5}.
\end{align*}
However, if $x(p_l) = D$, then $p_l$ must be on the base of a horizontal comb, whose predecessor and successor combs are vertical. Since the distance between the base and any vertical comb is at least $D-(cN+1) > 8d+1$, the length of $\pi$ would be larger than $8d+1$, a contradiction. Hence $x(p_l) \neq D$. Similarly, $y(p_l) \neq D$.

The discussion above implies that for all $l$, we have $p_l = (ci_l+\delta_l, cj_l+\theta_l, z_l)$ for some $i_l, j_l \in \set{1,\dots,N}$, $\delta_l,\theta_l \in \set{0,1}$, and $z_l \in \set{0,\dots,5}$. Since $\pi$ has length at most $8d+1 < c$, the values $i_1, \dots, i_k$ must coincide, and the values $j_1, \dots, j_k$ must coincide. In other words, $p_l = (ci+\delta_l, cj+\theta_l,z_l)$ where $i, j \in \set{1,\dots,N}$ do not depend on $l$.

For $0 \leq t \leq d$, we have
\begin{itemize}
    \item $p_{4t+1} \in V$, and thus $p_{4t+1} = (ci, cj+\theta_{4t+1}, 3)$;
    \item $p_{4t+2} \in H$, and thus $p_{4t+2} = (ci+\delta_{4t+2}, cj, 2)$.
\end{itemize}
Regardless of the actual values of $\theta_{4t+1}, \delta_{4t+2}$, we can always route the subtour between $p_{4t+1}$ and $p_{4t+2}$ to pass through the points $(ci, cj, 3)$ and $(ci, cj, 2)$ without changing its length. Since $(ci,cj,3) \in V$ and $(ci,cj,2) \in H$, we can replace $p_{4t+1} = (ci,cj,3)$ and $p_{4t+2} = (ci,cj,2)$ for all $0 \leq t \leq d$. Hence the previous claim implies that the length of $\pi$ is \emph{at least} $8d + 1$.

Since we assumed that the length of $\pi$ is at most $8d + 1$, all subtours given by the previous claim must have the smallest possible length, namely $7$. Hence $(a_i[t],b_j[t]) \neq (1,1)$ for all $t$, or in other words $a_i,b_j$ is a pair of orthogonal vectors.
\end{claimproof}

Suppose that \touringOrthoThreeParam can be solved in time $O(n^{2-\eps} \poly(k))$ for some $\eps > 0$. Then in particular we can solve the constructed instance in time $O((kN)^{2-\eps} \poly(k)) = O(N^{2-\eps} \poly(d))$, which violates OVH.
\end{proof}

\section*{Acknowledgement}
This work was initiated at the 2025 Lorentz Center workshop ``Fine-Grained and Parameterized Computational Geometry'' in Leiden, the Netherlands.

\bibliography{references}

\clearpage

\appendix

\section{Rectangle partition with low stabbing number}

\begin{proof}[Proof of \Cref{lem:rect-partition}]
    Let $V$ be the set of vertices of $P_1, \dots, P_k$. We build a quadtree on $V$ as follows. Let $B$ be the bounding box of $V$. We split $B$ into two rectangles $B_0, B_1$ by a vertical line at the median $x$-coordinate in $V$. Then we split each $B_i$ into two rectangles $B_{i 0}, B_{i 1}$ by a horizontal line at the median $y$-coordinate in $V \cap B_i$. We recurse this process on every rectangle until the interiors are free of points.

    Since each split kills one point from the interior and creates only one more rectangle, the total number of rectangles is $n + 1$. Moreover, no rectangle contains a vertex in $V$.

    Next we argue that every vertical/horizontal line intersects at most $O(\sqrt{n})$ rectangles. Without loss of generality, we consider a horizontal line $l$ that does not hit $V$. (If $l$ hits $V$, then we take a line $l^+$ slightly above and another line $l^-$ slightly below. Every rectangle intersected by $l$ must be intersected by either $l^+$ or $l^-$, so it suffices to consider the latter two.) At the beginning of the process, $l$ intersects the bounding box $B$. After the first $x$ and $y$-splits, $l$ is broken into two parts, each intersecting the interior of just one rectangle $B_{ij}$. Note that $B_{ij}$ contains only a quarter of the points in its parent $B$. As we recursively split $B_{ij}$, the total number of rectangles stabbed by $l$ is bounded by the recursion $g(n) \leq 2 g(n/4)$ which  gives $g(n) = O(\sqrt{n})$.
\end{proof}

\end{document}